\documentclass[10pt]{article}
\usepackage{cite}
\usepackage{graphicx}
\usepackage{siunitx}
\usepackage{mathtools}
\usepackage{authblk}

\DeclarePairedDelimiter\floor{\lfloor}{\rfloor}
\usepackage{amsmath}
\usepackage{amsthm}
\usepackage{amssymb}
\usepackage{dsfont}
\usepackage{epsfig}
\usepackage{color}
\newtheorem{proposition}{Proposition}
\newtheorem{lemma}{Lemma}
\newtheorem{corollary}{Corollary}
\newtheorem{theorem}{Theorem}
\newcommand{\defeq}{\mathrel{\mathop:}=}

\begin{document}
\title{Asymptotic Approximations for TCP Compound}
\author{Sudheer Poojary \qquad Vinod Sharma \\ Department of ECE, Indian Institute of Science, Bangalore \\ Email: \{sudheer, vinod\}@ece.iisc.ernet.in}
\maketitle

\begin{abstract}
\label{sec:abstract}
In this paper, we derive an approximation for throughput of TCP Compound connections under random losses.  Throughput expressions for TCP Compound under a deterministic loss model exist in the literature. These are obtained assuming the window sizes are continuous, i.e., a fluid behaviour is assumed. We validate this model theoretically. We show that under the deterministic loss model, the TCP window evolution for TCP Compound is periodic and is independent of the initial window size. We then consider the case when packets are lost randomly and independently of each other. We discuss Markov chain models to analyze performance of TCP in this scenario. We use insights from the deterministic loss model to get an appropriate scaling for the window size process and show that these scaled processes, indexed by $p$, the packet error rate, converge to a limit Markov chain process as $p$ goes to $0$. We show the existence and uniqueness of the stationary distribution for this limit process. Using the stationary distribution for the limit process, we obtain approximations for throughput, under random losses, for TCP Compound when packet error rates are small. We compare our results with ns2 simulations which show a good match. 
\end{abstract}
\section{Introduction}
\label{sec:introduction}
% Motivation for high speed TCP
In the last few years, traditional TCP congestion control algorithms, viz., TCP Reno\cite{rfc2581}, TCP New Reno\cite{rfc2582} are being superseded by high speed TCP variants \cite{Yang2014}. The reason for this is the poor link utilization by traditional TCP variants over high-speed, large delay networks and networks subject to random losses \cite{Huston2006, rfc3649}. Traditional TCP variants,  use AIMD algorithm with fixed rates of increase and rates of decrease.  For high bandwidth delay product links, these rates are too conservative and lead to inefficient utilization. Also, TCP cannot distinguish between congestion and non-congestion losses, hence it has poor performance over wireless links. Since the high speed TCP variants use more aggressive increase rates, they perform more efficiently than the traditional TCP on high speed, large delay networks. Also, as a consequence of their aggressive increase rates, these protocols are able to recover quickly from a non-congestion loss and therefore perform better than traditional TCP in wireless environments. 

%Examples of high speed TCP algorithms
High speed TCP\cite{rfc3649}, H-TCP\cite{Leith2004}, BIC\cite{Xu2004}, CUBIC\cite{Ha2008}, Scalable\cite{Kelly2003}, FAST\cite{Wei2006TON} and Compound\cite{Tan2006Infocom} are some examples of high speed TCP congestion control algorithms. High speed TCP uses AIMD algorithm but with variable increase and decrease factors which are aggressive when window sizes are large and conservative when window sizes are small. BIC uses binary search to search for the optimum sending rate. H-TCP and CUBIC adjust their window sizes based on time elapsed since last congestion and increase the window size more rapidly as this time increases. Scalable TCP uses MIMD congestion control which is more aggressive than AIMD. FAST and Compound are delay-based TCP algorithms. The TCP Compound window has two components, a delay-based component and a loss-based component. When competing with loss-based congestion control algorithms, delay-based protocols may not be able to get their fair share of network capacity. The loss-based component of TCP Compound ensures that these flows get a fair share of the network capacity.

% Literature Survey
We now give a brief overview of literature on analytical models for TCP performance analysis.  There exists enormous literature for analysis of the traditional variants of TCP. In \cite{Padhye2000}, the authors derive an expression for the steady state throughput of long lived TCP Reno flows as a function of the loss rates and the RTT of the flows. In \cite{Lakshman1997}, the authors study the performance of TCP Reno on high bandwidth delay product networks under random losses. They find that when multiple flows share a bottleneck queue, connections with higher round trip times receive lesser throughput. In \cite{Cardwell2000}, the authors study the effect of connection establishment, slow start and losses on the TCP Reno data transfer latencies. In \cite{Gupta2006}, the authors analyze performance of persistent and ON-OFF TCP Reno connections through a bottleneck queue with UDP traffic and compute mean file download times. In \cite{Sharma2004}, the authors prove stability of TCP Tahoe and Reno with drop-tail and RED queues in tandem which also carry UDP traffic. In \cite{Shorten2007}, the authors use a random matrix model to describe the performance of multiple AIMD TCP flows through a drop-tail queue. In \cite{Shakkottai2002}, the authors model congestion control algorithms using deterministic fluid model approximations for a system with large number of flows. In \cite{Misra2000}, the authors develop a fluid model for window size evolution of TCP Reno flows going through a RED queue.  

%Literature Survey on Compound TCP
There are fewer analytical studies of TCP Compound as compared to the traditional AIMD variants. In \cite{Ghosh2014, Chavan2015, Manjunath2015}, the authors study the performance of TCP Compound using control theoretic techniques and derive stability conditions for TCP Compound. In these papers, it is shown that when multiple TCP flows share a single bottleneck queue, the queue sizes and the link utilization has oscillatory behaviour when the feedback delays (round trip times of the flows) and the buffer sizes are large. In \cite{Ghosh2014}, the authors evaluates TCP Compound performance as a function of buffer size in the bottleneck queue. In \cite{Chavan2015}, the authors study the performance of TCP Compound with a proportional integral enhanced queue management policy whereas in \cite{Manjunath2015}, RED queue management policy is considered. In \cite{Blanc2009}, the authors study the performance of a single TCP Compound connection in the presence of random losses using Markovian models. In \cite{Sudheer2013}, we develop Markov models for TCP Compound under random losses, also considering the case when queuing delays are non-negligible. These results were then used to compute TCP Compound performance in a heterogeneous network with multiple queues when competing against different TCP variants (TCP CUBIC and TCP Reno). 

% Our contribution
We now briefly describe the contribution and the motivation for this paper. Throughput expressions for Compound connections have been derived earlier using a deterministic loss model in \cite{Tan2006Infocom}. Throughput for TCP Compound connections has also been computed using Markov chains in \cite{Sudheer2013, Blanc2009}. The deterministic loss model gives us an explicit expression, whereas the Markov chain models only give numerical results. In this paper, we address this drawback of the Markov chain model and give an approximation for TCP throughput under random losses. We first discuss the deterministic loss model and then give theoretical justification for it. These models rely on the existence of a globally stable fixed point for equation \eqref{eqn:fixed_point_form} discussed in Section \ref{sec:fluidmodel}. While \cite{Tan2006Infocom} implicitly shows that a unique fixed point exists for \eqref{eqn:fixed_point_form} for TCP Compound, we show that the fixed point is globally stable. This implies that under the deterministic loss model, the asymptotic TCP window evolution is periodic and does not depend on the initial window size. For the Markov chain model, we consider appropriately scaled window size processes for TCP and derive results for the limiting process (as the packet error rates, $p \rightarrow 0$) and develop an approximation for TCP throughput for small packet error rates. The deterministic loss model is used to get the correct scaling. This approach has earlier been used in \cite{Dumas2002} where they develop an approximation for TCP throughput for TCP Reno. While the earlier Markov chain models in \cite{Sudheer2013, Blanc2009} requires solving for the stationary distribution of the window size process for each $p$, the approach discussed in this paper requires us to solve for the stationary distribution of a single Markov chain.
% Why TCP Compound?
We choose TCP Compound for analysis as it is a widely used congestion control algorithm. In \cite{Yang2014}, it has been observed that out of $30000$ servers chosen, $14.6-25.6\%$ used TCP Compound. Also, TCP Compound is available as an option on Microsoft Windows systems and is used by Microsoft Windows servers.

The organization of the paper is as follows. In Section \ref{sec:systemmodel}, we describe our system model. In Section \ref{sec:fluidmodel}, we describe the deterministic loss model and show that the window evolution under this model is eventually periodic. In Section \ref{sec:markovchain}, we consider a scaled window size process. We show that with appropriate scaling, starting from a window size $x$, the limit for the scaled time between losses, as $p \rightarrow 0$ converges to a random variable. We then show that the sequence of the limiting window size process at drop epochs is a Markov chain with a unique stationary distribution. The stationary distribution of these limit processes is then used to derive an approximation for the TCP Compound throughput for small packet error rates. In Section \ref{sec:simulation_results}, we compare our approximation with ns2 simulations. Section \ref{sec:conclusion} concludes the paper

\section{System model for TCP Compound}
\label{sec:systemmodel}
We illustrate our system model in Figure \ref{fig:singleTCP}. We consider a single TCP connection with constant RTT (round trip time), i.e., negligible queuing with random losses. A packet in a TCP window can be dropped with probability $p$ independently of other packets. This is a commonly made assumption, especially for wireless links \cite{Padhye2000}, \cite{Lakshman1997}. We will compute an approximation for TCP Compound average window size under these assumptions in Section \ref{sec:markovchain}.
\begin{figure}
\centering
\includegraphics[scale=0.6]{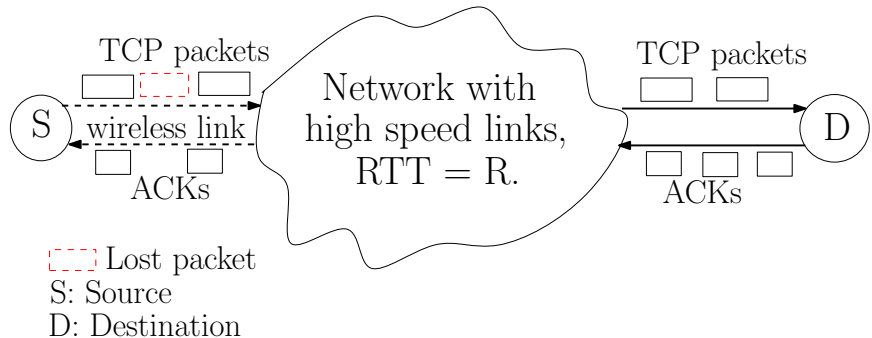}
\caption{Single TCP with fixed RTT}
\label{fig:singleTCP}
\end{figure}

The TCP Compound window size has two components: a loss-based component and a delay based component. Let $W_n$ denote the window size of TCP Compound at the end of $n^{th}$ round trip time (RTT). Let $D_n$ and $L_n$ denote the delay-based and loss-based components respectively, where $W(n)=$ $L_n + D_n$. Then the TCP Compound window size evolution is given by 
\begin{equation}
\centering
\label{eqn:Dn}
D_{n+1} = 
\begin{cases}
D_{n} + (\alpha (W_{n})^{k} - 1)^{+}, \\ \hspace*{0.5cm}\text{ if no loss in $n^{th}$ RTT and } Q_{n+1} < \gamma; \\
(D_{n} - \zeta Q_{n+1})^{+}, \\ \hspace*{0.5cm} \text{ if no loss in $n^{th}$ RTT and } Q_{n+1} \geq \gamma; \\
(1 - \beta) D_{n}, \text{ if loss is detected};
\end{cases}
\end{equation}

\begin{equation}
\centering
\label{eqn:Ln}
L_{n+1} = 
\begin{cases}
L_{n} + 1, \text{ if no loss in $n^{th}$ RTT}; \\
(1 - \beta) L_{n}, \text{ if a loss is detected;}
\end{cases}
\end{equation}
here $\alpha$, $\beta$ and $k$ are TCP Compound parameters, whereas $\gamma$ is a queuing threshold. When there is no queuing, the delay-based component is more aggressive as compared to TCP Reno and helps TCP Compound flows achieve higher utilization over high speed links. When there is queuing, the loss-based component ensures a worst-case TCP Reno-like behaviour.

In Section \ref{sec:fluidmodel}, we briefly describe the approximation for TCP Compound average window size under a deterministic loss model. The insights from the deterministic loss model are then used to derive the approximation for TCP Compound average window size under random losses in Section \ref{sec:markovchain}.
\section{Deterministic loss models for TCP}
\label{sec:fluidmodel}
Fluid models disregard the discrete nature of TCP window size and treat the TCP window size as a fluid, i.e., the window size now takes values in $(0,\infty)$. Such an assumption allows us to use tools such as differential equations which makes analysis of fluid models simpler. The reference \cite{Mills2011} gives a survey of deterministic fluid models which compute the TCP response function. The response function for TCP is an expression for TCP throughput as a function of the system parameters viz., RTT, packet error rate. Such approximations have been used by \cite{Ha2008}, \cite{Tan2006Infocom} and \cite{Mathis1997} to get expressions for average window size for TCP CUBIC, Compound and Reno respectively. We briefly discuss these models and then show that under the deterministic loss model, the window size evolution is asymptotically periodic and independent of the initial window size. This validates the TCP Compound throughput approximation in \cite{Tan2006Infocom}. We use the insights from the deterministic loss model to appropriately scale the TCP Compound window evolution process and derive an approximation for throughput under random losses in Section \ref{sec:markovchain}. 

\subsection{Fluid model description}
\label{sec:fluidmodel_description}
Consider a single TCP connection with RTT $R$, and suppose the packets are dropped independently of each other with probability $p$. The average number of packets sent between two consecutive losses is $\frac{1}{p}$. Let $W(t)$ be the window size at time $t$. We now consider a deterministic process, $\hat{W}(t)$ to approximate the TCP window evolution. The evolution of $\hat{W}(t)$ process is similar to the evolution of the $W(t)$ process between losses. However, unlike $W(t)$, $\hat{W}(t)$ is not integer-valued. Also the process $\hat{W}(t)$ is subject to losses every $\frac{1}{p}$ packets and hence undergoes a reduction every $\frac{1}{p}$ packets. Suppose $\hat{W}(0) = x$, let $\tau_p(x)$ denote the time taken by the $\hat{W}(t)$ process to send $\frac{1}{p}$ packets. At $t = \tau_p(x)$, $\hat{W}(t)$ undergoes a window reduction and $\hat{W}(\tau_p(x)^+) =  (1 - \beta)\hat{W}(\tau_p(x))$, where $\beta$ is the multiplicative drop factor for TCP. Next, the window size $\hat{W}(t)$ evolves as before but now with initial window size, $\hat{W}(\tau_p(x)^+)$. Now at time $t = \tau_p(x) + \tau_p(\hat{W}(\tau_p(x)^+))$, $\hat{W}(t)$ process undergoes another loss reducing its window size. This process continues so that $\hat{W}(t)$ undergoes losses every $\frac{1}{p}$ packets and in between the loss epochs the evolution is similar to $W(t)$ process without losses. Thus the deterministic loss model has a loss rate of $p$.

Suppose there exists a $x^*_p$ such that $\hat{W}(\tau_p(x^*_p)^+) = x^*_p$, i.e., the fixed point equation
\begin{equation}
\label{eqn:fixed_point_form}
\hat{W}(\tau_p(x)^+) = x,
\end{equation}
has a unique solution. Then, if we start from $x^*_p$, the process $\hat{W}(t)$ will have a periodic behaviour with period $\tau_p(x^*_p)$ and $\hat{W}(t) \in [x^*_p, \frac{x^*_p}{1-\beta}]$.  The long time average for the process $\hat{W}(t)$ is then given by
\begin{equation}
\label{eqn:generic_rf}
 \frac{1}{\tau_p(x^*_p)} \int_{0}^{\tau_p(x^*_p)}  \hat{W}(t) dt,
\end{equation}
with $\hat{W}(0) = x^*_p$. Expression \eqref{eqn:generic_rf}, is then used to approximate the time stationary window size $E[W(p)]$ of corresponding window size $W(t)$ process. Using the above model, \cite{Tan2006Infocom} derives the following approximations for average window size for TCP Compound.
\begin{equation}
\label{eqn:compound_rf}
E[W(p)] \approx \frac{(1 - k) (\frac{\alpha}{p})^{\frac{1}{2-k}} (1 - (1 -\beta)^{(2-k)})^{\frac{1-k}{2-k}}}{(1 - (1 -\beta)^{(1-k)}) (2 - k)^{\frac{1-k}{2-k}}}.
\end{equation}
The throughput of the TCP connection is given by $\frac{E[W(p)]}{R}$.

The above expression is valid, for any initial condition at $t=0$, if the fixed point equation, \eqref{eqn:fixed_point_form} is globally stable. We prove this for TCP Compound in Proposition \ref{prop:CTCP}. We ignore the slow start phase and ignore that there may be an upper bound on the maximum window size. This assumption is implicitly made in the references \cite{Ha2008}, \cite{Tan2006Infocom} and \cite{Mathis1997}.

\subsection{TCP Compound Fluid model}
\label{sec:tcpCompoundFluid}
For TCP Compound, the window evolution, when there is no loss, with fixed RTT and no queuing delay is given by,
\begin{equation}
\label{eqn:CTCP_fixedRTT}
W(t + R) = W(t) + \alpha W(t)^k,
\end{equation}
where $R$ is the RTT for the connection. Therefore, for the deterministic fluid model $\hat{W}(t)$, the window size between losses can be approximated as solution to the following ODE,
\begin{equation}
\label{eqn:CTCP_ode}
\frac{d\hat{W}}{dt} = \frac{\alpha \hat{W}^k}{R},
\end{equation}
which is 
\begin{equation}
\frac{\hat{W}(t)^{1-k}}{1-k} = \frac{\alpha t}{R} + c,
\end{equation}
where c depends on the initial conditions of \eqref{eqn:CTCP_ode}. For the process, $\hat{W}(t)$, we have  
\begin{proposition}
\label{prop:CTCP}
For any fixed $p \in (0,1)$, there exists a unique $x$ (denoted by $x^*_p$) such that $\hat{W}(\tau_p(x)^+) = x$, i.e., the function $\hat{W}(\tau_p(.)^+)$ has a unique fixed point. For any $x \geq 1$ such that $\hat{W}(0) = x$, $\hat{W}(t)$ converges monotonically to $x^*_p$ at drop epochs. 
\end{proposition}
\begin{proof}

\subsubsection*{Existence of $x^*_p$}
Suppose that at time $t = 0$, the flow experienced a loss and $\hat{W}(0^+) = x$. Then,
\begin{equation}
\label{eqn:CTCP_W}
\frac{\hat{W}(t)^{1-k}}{1-k} = \frac{\alpha t}{R} + \frac{\hat{W}(0^+)^{1-k}}{1-k},
\end{equation}
assuming no further losses in $(0,t]$. Let $\tau_p(x)$ be the first drop epoch after $t = 0$, so that $\frac{1}{p}$ packets are sent by the $\hat{W}(t)$ process in $(0,\tau_p(x)]$. Then we have 
\begin{equation}
\label{eqn:CTCP_eqn1}
\begin{aligned}
\frac{1}{p} &= \frac{1}{R}\int_{0}^{\tau_p(x)} \hat{W}(t) dt \\
&= \frac{1}{R}\int_{t = 0}^{\tau_p(x)} \hat{W}(t)^{1-k} \frac{R}{\alpha} d\hat{W},  \hspace*{1cm} \text{from \eqref{eqn:CTCP_ode},} \\
&= \frac{\hat{W}(\tau_p(x))^{2-k} - x^{2-k}}{\alpha (2-k)}.
\end{aligned}
\end{equation}
Solving the fixed point equation \eqref{eqn:fixed_point_form} now requires us to find an $x$ such that
\begin{equation}
\label{eqn:fixed_point_form_ctcp}
\hat{W}(\tau_p(x)^+) = x = \hat{W}(0^+),
\end{equation}
so that $\hat{W}(\tau_p(x)) = \frac{x}{1-\beta}$. From \eqref{eqn:CTCP_eqn1} and \eqref{eqn:fixed_point_form_ctcp}, for the fixed point, $x_p^*$, we have
\begin{equation}
\label{eqn:CTCP_eqn2}
\frac{1}{\alpha (2 -k) } \Bigl( \Bigl( \frac{x^*_p}{1 - \beta } \Bigr)^{2-k} - (x^*_p)^{2-k} \Bigr) = \frac{1}{p}.
\end{equation}
Thus the fixed point to \eqref{eqn:fixed_point_form} for TCP Compound fluid model is given by
\begin{equation}
\label{eqn:CTCP_FP_exist}
x^*_p = \Bigl( \frac{\alpha (2 - k) (1 - \beta)^{2-k}}{p(1 - (1 - \beta)^{2-k})} \Bigr)^{\frac{1}{2-k}}.
\end{equation}
Thus, corresponding to every $p$, we have unique $x$ given by \eqref{eqn:CTCP_FP_exist} such that $\hat{W}(\hat{\tau}_p(x)) = \frac{x}{1-\beta}$. Thus we have proved existence of $x^*_p$ such that $\hat{W}(\tau_p(x^*_p)^+) = x^*_p$.

\subsubsection*{Convergence to $x^*_p$}
From \eqref{eqn:CTCP_eqn1} and \eqref{eqn:CTCP_eqn2}, we have
\begin{equation}
\label{eqn:CTCP_eqn5}
\hat{W}(\tau_p(x))^{2-k} = (x^*_p)^{2-k} \Bigl(\Bigl(\frac{1}{1-\beta}\Bigr)^{2-k} - 1 \Bigr) + x^{2-k}.
\end{equation}
Suppose the initial window size, $\hat{W}(0^+)$, at the first loss epoch at $t=0$ is $x$, such that  $0 < x < x^*_p$, then $\frac{x}{1-\beta} < \hat{W}(\tau_p(x)) < \frac{x^*_p}{1 - \beta}$. Thus after window reduction at $\tau_p(x)^+$, we have $ x <  \hat{W}((\tau_p(x))^+) < x^*_p$. Therefore if $\hat{W}(0^+) = x < x^*_p$, $\hat{W}$ monotonically increases to $x^*_p$ at the drop epochs. Similarly, if $x > x^*_p$, then $x^*_p < \hat{W}((\tau_p(x))^+) < x$. Thus, if $\hat{W}(t) = x $ at some drop epoch, the value of $\hat{W}(t)$ at drop epochs converges monotonically to $x^*_p$.
\end{proof}
In Figure \ref{fig:det_evolution_CTCP}, we illustrate the evolution of the window size process, $\hat{W}(t)$, under the deterministic loss model for packet error rate, $p= 0.001$ with initial window sizes, $10$, $100$. In both cases, we see that the window size at the drop epoch converges monotonically to $x^*_p$ and the window size has a periodic evolution eventually.
\begin{figure}
\centering
\includegraphics[scale=0.20, trim = 100 20 0 0, clip=true]{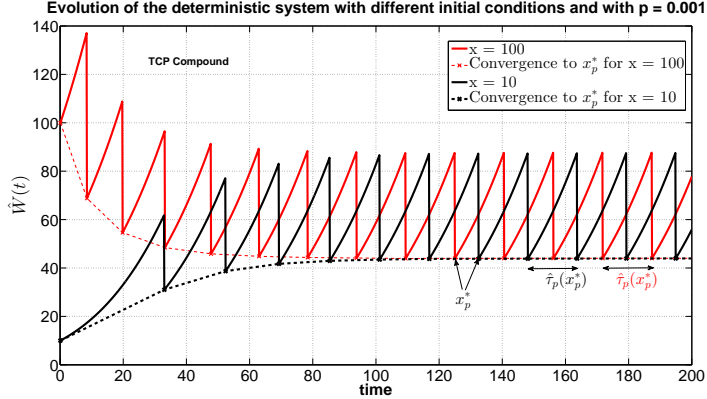}
\caption{Deterministic-loss window evolution with initial conditions, $10, 100$}
\label{fig:det_evolution_CTCP}
\end{figure}

Using equation \eqref{eqn:CTCP_W} in Proposition \ref{prop:CTCP}, for the deterministic loss model, time between losses converges to 
\begin{equation}
\label{eqn:ctcp_time_betn_losses}
\hat{\tau}_p(x^*_p) =  \Bigl( \frac{(1 - (1-\beta)^{1-k}) R}{\alpha (1 - k) (1-\beta)^{1-k}} \Bigr) (x^*_p)^{1-k}.
\end{equation}
From equation \eqref{eqn:CTCP_FP_exist}, we see that for the deterministic loss model, the window size at the loss epochs converges to $C_1 p^{-\frac{1}{2-k}}$, where $C_1$ is some constant dependent on the TCP Compound parameters. Also from equation \eqref{eqn:ctcp_time_betn_losses}, the time between losses converges to $C_2 p^{-\frac{1-k}{2-k}}$ where $C_2$ is some constant dependent on the TCP Compound parameters. These are key observations which we use in the next section to get appropriate scaling for the TCP Compound window size process under random losses.
\section{Throughput approximation with random losses}
\label{sec:markovchain}

In this section, we assume that the packets of the connection are subject to random losses due to channel errors. Also assume that each packet in a window is lost independently of other packets with probability $p$. We can then model the window evolution as a Markov chain. The results obtained using this Markov chain are presented in \cite{Sudheer2013}. In \cite{Sudheer2013}, we assume that the maximum window size is restricted by $W_{max} < \infty$. With a maximum window size restriction, it is easy to see that as $p \rightarrow 0$, the steady state distribution of the Markov chain $\{W_n\}$ converges to a unit mass on $W_{max}$.  In this paper, we derive limiting results for this Markov chain as $p$ goes to zero. We show that with an appropriate scaling and letting $W_{max} = \infty$, approximations for distributions for time between losses, window size at loss epochs and finally time average window size can be obtained for small $p$.

Consider a single TCP Compound connection with fixed RTT $R$, whose window evolution is given by \eqref{eqn:Dn} and \eqref{eqn:Ln}. Let $W_n(p)$ be the window size at the end of the $n^{th}$ RTT. Let $V_k(p)$ denote the window size at the $k^{th}$ loss epoch (just after loss) and let $G^p_{V_k(p)}$ denote the time between the $k^{th}$ and $(k+1)^{st}$ loss. In Proposition \ref{prop:ergodicity} and Corollary \ref{coro:ergodicity}, we show  that the processes $\{W_n(p)\}$, $\{V_k(p)\}$ and $\{G^p_{V_k(p)}\}$ have unique stationary distributions.

\begin{proposition}
\label{prop:ergodicity}
For any $p \in (0,1)$, the Markov chain $\{W_n(p)\}$ has a single aperiodic positive recurrent communicating class and the remaining states are transient. Hence it has a unique stationary distribution. Also, for fixed $p\in(0,p)$, $E[W^\eta]$ is finite under stationarity for all $\eta \geq 1$. 
\end{proposition}

\begin{proof}
The state $(1)$ can be reached from any state in the state space of $\{W_n(p)\}$ process with non-zero probability. Hence, all states that can be reached from the state $(1)$ form a communicating class, whereas the remaining states are transient. Also, the state $(1)$ has a self loop. Hence the communicating class which contains state $(1)$ is aperiodic. We next show that the communicating class containing state $(1)$ is positive recurrent.

For the Markov chain $\{W_n\}$, we have 
\begin{equation}
\label{eqn:ctcp_ergodicity}
\begin{aligned}
E[W_{n+1} - W_n | W_n = j] &= \frac{-j}{2} (1 - (1-p)^j) + \alpha j^k (1 - p)^j, \\
&= \frac{-j}{2} + (1 - p)^j (\alpha j^k + \frac{j}{2}), \\
&\leq \frac{-j}{2} + (1 - p)^j (\alpha  + \frac{1}{2}) j, \\
\end{aligned}
\end{equation}
For $p \in (0,1)$, we can choose a $J$ such that $\forall j > J$ we have $E[W_{n+1} - W_n | W_n = j] < -\epsilon \ $ for some $\epsilon > 0$. Thus by mean drift criteria for positive recurrence, the Markov chain $\{W_n\}$ is positive recurrent.

For proof of $E[W^{\eta}] < \infty$ for all $\eta \geq 1$, we will use Theorem $3$ in \cite{Tweedie1983}. Since $\{W_k\}$ is a countable state irreducible, aperiodic Markov chain, any finite set of the state space is small. Let us take $A = \{1, \cdots M \}$ where $M < \infty$ will be specified later. To show $E[W^{\eta}] < \infty$, it is sufficient to show that 

\begin{equation}
\label{eqn:oneTweedie}
\sup_{i \in A} E[W_1^\eta | W_0 = i] < \infty, 
\end{equation}
and there is a $\delta > 0$ such that
\begin{equation}
\label{eqn:twoTweedie}
E[W_1^\eta | W_0 = i] \leq (1 - \delta) i^\eta, 
\end{equation}
for all $i \in A^c$.
Since, 
\begin{equation}
\label{eqn:threeTweedie}
\begin{aligned}
E[W_1^\eta | W_0 = i] &= (i+\alpha i^k)^\eta (1 - p)^i + (\frac{i}{2})^\eta (1 - (1-p)^i) \\
& \leq i^\eta (1 + \alpha)^\eta + (\frac{i}{2})^\eta (1 - (1-p)^i) \\
& \leq M^\eta (1 + \alpha)^\eta + (\frac{M}{2})^{\eta} < \infty 
\end{aligned}
\end{equation}
for any $i \in A$, \eqref{eqn:oneTweedie} is proved. Also, for any $i > M$, $(1 - p)^i < \frac{1}{4 (1+\alpha)^\eta}$ if $M$ is taken large enough. Then, for $i > M$
\begin{equation}
\begin{aligned}
E[W_1^\eta | W_0 = i] &\leq i^\eta(1 + \alpha)^\eta(1 - p)^i + (\frac{i}{2})^\eta, \\
& \leq \frac{i^\eta}{4} + \frac{i^\eta}{2}
\end{aligned}
\end{equation}
Thus for $i >M$
\begin{equation*}
E[W_1^\eta | W_0 = i] \leq \frac{3}{4} i^{\eta}.
\end{equation*}
\end{proof}

\begin{corollary}
\label{coro:ergodicity}
For any $p \in (0,1)$, the processes $\{V_k(p)\}$ and $\{G^p_{V_k(p)}\}$ have unique stationary distributions. Also, for fixed $p\in(0,p)$, $E[V(p)^\eta]$ is finite under stationarity for all $\eta \geq 1$. 
\end{corollary}
\begin{proof}
For the process $\{W_n(p)\}$, consider the inter-visit times to state $1$. These epochs are regeneration epochs for the process $\{W_n(p)\}$ as well as for the $\{V_k(p)\}$ and $\{G^p_{V_k(p)}\}$ processes. From Proposition \ref{prop:ergodicity}, the process $\{W_n(p)\}$ is positive recurrent. Hence, with visits to state $1$  as regeneration epochs, the average regeneration cycle length, $\mathbb{E}\tau_{W}$, (for the $\{W_n(p)\}$ process) is finite. The regeneration cycle length for $\{V_k(p)\}$ and $\{G^p_{V_k(p)}\}$ processes (denoted by $\tau_{V}$) is given by the number of loss epochs between two consecutive visits to state $1$. Since in each regeneration cycle, $\tau_{V} \leq \tau_{W}$, we get, $\mathbb{E}\tau_{V} < \infty$. Therefore, the Markov chain $\{V_k(p)\}$ is positive recurrent and the processes $\{V_k(p)\}$ and $\{G^p_{V_k(p)}\}$  have unique stationary distributions. Also, since the process $\{V_k(p)\}$ is stochastically dominated by $\{W_n(p)\}$ process, the finiteness of $E[V(p)^\eta]$, under stationarity, follows from Proposition \ref{prop:ergodicity}.
\end{proof}

We now derive an approximation for the TCP throughput under random losses. Consider a single TCP Compound connection with fixed RTT $R$, whose window evolution is given as solution to the ODE \eqref{eqn:CTCP_ode} given in Section \ref{sec:fluidmodel}. The solution to this ODE with initial window size $x$ is given by 
\begin{equation}
\label{eqn:CTCP_Wt}
W(t) = \Bigl(x^{1-k} + \frac{\alpha (1 - k) t}{R}\Bigr)^\frac{1}{1-k}.
\end{equation}
Let $G_x^p$ denote the time for first packet loss (in multiples of $R$) when the initial window size is $x$ and the packet error probability is $p$. Consider $p^{\frac{1-k}{2-k}} G_{\floor{\frac{x}{p^{\frac{1}{2-k}}}}}^p$, which denotes the product of  $p^{\frac{1-k}{2-k}}$ with time for first packet loss when the initial window size is $\floor{\frac{x}{p^{\frac{1}{2-k}}}}$. The choice of parameters $p^{\frac{1-k}{2-k}}$ and $p^{\frac{1}{2-k}}$ is motivated from the deterministic loss model in Section \ref{sec:fluidmodel}. In Section \ref{sec:fluidmodel}, it was shown that under deterministic losses which happen every $\frac{1}{p}$ packets, the average time between losses is inversely proportional to $p^{\frac{1-k}{2-k}}$  whereas the time average window size is inversely proportional to $p^{\frac{1}{2-k}}$ (see equations \eqref{eqn:CTCP_FP_exist} and \eqref{eqn:ctcp_time_betn_losses}). In Proposition \ref{prop:Gbarx_CTCP1}, we show that the term   $p^{\frac{1-k}{2-k}} G_{\floor{\frac{x}{p^{\frac{1}{2-k}}}}}^p$ converges to a random variable $\overline{G}_x$ as $p \rightarrow 0$ for all $x \geq 1$.

\begin{proposition}
\label{prop:Gbarx_CTCP1}
For $x \geq 1$, $p^{\frac{1-k}{2-k}} G_{\floor{\frac{x}{p^{\frac{1}{2-k}}}}}^p$ converges in distribution to a random variable $\overline{G}_x$, where the complementary cdf of $\overline{G}_x$ is given by
\begin{equation}
\label{eqn:Gbarx_CTCP2}
\begin{split}
\mathbb{P}(\overline{G}_x \geq y) = \exp \Bigl(-   & xy - 2\alpha (1 -k) x^k y^2 - 2 \alpha^2 (1-k)^2 x^{2k-1} y^3 \\
& - \alpha^3 (1 -k)^3 x^{3k-2} y^4 - \frac{\alpha^4 (1-k)^4 x^{4k-3} y^5}{5} \Bigr).
\end{split}
\end{equation}
For any finite $M$, if $x, y \leq M$ the above convergence is uniform in $x$ and $y$, i.e.,
\begin{equation}
\label{eqn:Gbarx_CTCP3}
\lim_{p \rightarrow 0} \sup_{ p^{\frac{1}{2-k}} \leq x \leq M, y \leq M } \Bigl|\mathbb{P}\Bigl(p^{\frac{1-k}{2-k}}G_{\floor{\frac{x}{p^{\frac{1}{2-k}}}}}^p \geq y\Bigr) - \mathbb{P}(\overline{G}_x \geq y) \Bigr| = 0,
\end{equation}
Also, the moment generating functions of $p^{\frac{1-k}{2-k}}G_{\floor{\frac{x}{p^{\frac{1}{2-k}}}}}^p$  are uniformly bounded in some neighbourhood of $0$.
\end{proposition}

\begin{proof}
We have,
\begin{equation}
\label{eqn:Gbarx_CTCP1_start}
%\begin{split}
\mathbb{P}\bigl(G_{\floor{\frac{x}{p^{\frac{1}{2-k}}}}}^p > \floor{ \frac{y}{p^{\frac{1-k}{2-k}}}}\bigr) \leq \mathbb{P}\bigl(p^{\frac{1-k}{2-k}}G_{\floor{\frac{x}{p^{\frac{1}{2-k}}}}}^p \geq y\bigr) \leq \mathbb{P}\bigl(G_{\floor{\frac{x}{p^{\frac{1}{2-k}}}}}^p \geq \floor{ \frac{y}{p^{\frac{1-k}{2-k}}}}\bigr), \\
%&= \exp\Bigl(-p(x_0 + x_1 + \cdots + x_{\floor{\frac{y}{p^{\frac{1-k}{2-k}}}}})\Bigr),
%\end{split}
\end{equation}
with 
\begin{equation}
\mathbb{P}\bigl(G_{\floor{\frac{x}{p^{\frac{1}{2-k}}}}}^p > \floor{ \frac{y}{p^{\frac{1-k}{2-k}}}}\bigr) = (1-p)^{x_0 + x_1 + \cdots + x_{\floor{\frac{y}{p^{\frac{1-k}{2-k}}}}}},
\end{equation}
where $x_0 = \floor{\frac{x}{p^{\frac{1}{2-k}}}}$ is the initial window size and $x_i$ is the window size at the end of the $i^{th}$ RTT. Using \eqref{eqn:CTCP_Wt}, we get $x_i =  \floor{x_0\bigl( 1 + \frac{\alpha (1 - k) i}{x_0^{1 - k}}\bigr)^{\frac{1}{1-k}}}$. Now $x_0\bigl( 1 + \frac{\alpha (1 - k) i}{x_0^{1 - k}}\bigr)^{\frac{1}{1-k}} \in [x_i, x_i+1] $. Let $m = \floor{\frac{y}{p^{\frac{1-k}{2-k}}}}$ and $z = \frac{\alpha (1-k)}{x_0^{1-k}}$.  Therefore,
\begin{equation}
\label{eqn:Gbarx_CTCP4}
\begin{split}
\mathbb{P}\bigl( G_{\floor{\frac{x}{p^{\frac{1}{2-k}}}}}^p > \floor{ \frac{y}{p^{\frac{1-k}{2-k}}}}\bigr) 
\in \bigl[& (1-p)^{\bigl(x_0(1 + (1 + z)^{\frac{1}{1-k}} + \cdots + (1 + mz)^{\frac{1}{1-k}}) + (m+1) \bigr)}, \\
 & (1-p)^{\bigl(x_0(1 + (1 + z)^{\frac{1}{1-k}} + \cdots + (1 + mz)^{\frac{1}{1-k}}) - (m+1) \bigr)} \bigr],
\end{split}
\end{equation}
Since $k = \frac{3}{4}$ for TCP Compound, we get,
\begin{equation}
\label{eqn:Gbarx_CTCP5}
\begin{split}
(1 + z)^{\frac{1}{1-k}}  +  \cdots + (1 + mz)^{\frac{1}{1-k}} 
=&  (1 + z)^4 + \cdots + (1 + mz)^4, \\
=&  (1 + 4z + 6z^2 + 4z^3 +z^4) + \cdots + (1 \\ 
 &  + 4mz + 6(mz)^2 + 4(mz)^3 + (mz)^4).
\end{split}
\end{equation}
We can simplify terms in \eqref{eqn:Gbarx_CTCP4} using 
\begin{equation}
\label{eqn:Gbarx_CTCP6}
\begin{split}
(1 + & 4z +  6z^2  + 4z^3 +z^4) +  \cdots + (1 + 4mz + 6(mz)^2 + 4(mz)^3 + (mz)^4)) \\
& = m + 4z \sum_{i=1}^{m} i + 6z^2 \sum_{i=1}^{m} i^2 + 4z^3 \sum_{i=1}^{m} i^3 + z^4 \sum_{i=1}^{m} i^4. \\
\end{split}
\end{equation}
Therefore,
\begin{equation}
\label{eqn:Gbarx_CTCP7}
\begin{split}
x_0 & (1 + (1 + z)^{\frac{1}{1-k}} + \cdots + (1 + mz)^{\frac{1}{1-k}})  \\
=&  x_0 \bigl( (m +1)+ 4z \sum_{i=1}^{m} i + 6z^2 \sum_{i=1}^{m} i^2 + 4z^3 \sum_{i=1}^{m} i^3 + z^4 \sum_{i=1}^{m} i^4 \bigr). 
\end{split}
\end{equation}
After we expand the series, $\sum_{i = 1}^{m} i$, $\sum_{i = 1}^{m} i^2$,  $\sum_{i = 1}^{m} i^3$ and  $\sum_{i = 1}^{m} i^4$, we see that the terms in the RHS of equation \eqref{eqn:Gbarx_CTCP7} are of the form $ x_0 m^n z^j$ with $n - 1 \leq j \leq 4$, $1 \leq n \leq 5$. Also, we have an additional term $x_0$. Now,
\begin{equation}
\label{eqn:Gbarx_CTCP8}
x_0 m^n z^j = \alpha^j (1 - k)^j x^{1 - j + kj} y^n  p^{\frac{-1+(j-n)(1-k)}{2-k}}.
\end{equation}
If $j = (n-1)$, we have $\lim_{p \rightarrow 0} (1-p)^{p^{\frac{-1+(j-n)(1-k)}{2-k}}}$ $= \lim_{p \rightarrow 0} (1-p)^{\frac{1}{p}}$ $= e^{-1}$. If $j > n-1$, we get, $(1-p)^{p^{\frac{-1+(j-n)(1-k)}{2-k}}}$ $= (1-p)^{\frac{t-4}{5}}$, where $t = (j - n) \geq 0$. We have $\frac{t-4}{5} > -1$, for all $t \geq 0$ and $\lim_{p \rightarrow 0} (1-p)^{p^f} = 1$, if $f > -1$. Therefore, $\lim_{p \rightarrow 0} (1-p)^{p^{\frac{-1+(j-n)(1-k)}{2-k}}}$ $= 1$, for $j > n-1$ . Also, for the terms $p(m+1)$ and $x_0$, $\lim_{p \rightarrow 0} (1-p)^{(m+1)} =$ $\lim_{p \rightarrow 0} (1-p)^{-(m+1)} =  1$ and $\lim_{p \rightarrow 0} (1-p)^{x_0} = 1$. Thus, for the limit of equation \eqref{eqn:Gbarx_CTCP4} as $p \rightarrow 0$, we need to only consider terms of the form $x_0 m^n z^j$ with $j = n-1$. Therefore, we get
\begin{equation}
\label{eqn:Gbarx_CTCP9}
\begin{split}
\mathbb{P}\bigl( G_{\floor{\frac{x}{p^{\frac{1}{2-k}}}}}^p > \floor{ \frac{y}{p^{\frac{1-k}{2-k}}}}\bigr) \rightarrow \exp \bigl(-   & xy - 2\alpha (1 -k) x^k y^2 - 2 \alpha^2 (1-k)^2 x^{2k-1} y^3 \\
& - \alpha^3 (1 -k)^3 x^{3k-2} y^4 - \frac{\alpha^4 (1-k)^4 x^{4k-3} y^5}{5} \bigr),
\end{split}
\end{equation}
as $p \rightarrow 0$.
Using similar steps as before, we can show that
\begin{equation}
\label{eqn:Gbarx_CTCP9A}
\begin{split}
\mathbb{P}\bigl( G_{\floor{\frac{x}{p^{\frac{1}{2-k}}}}}^p \geq \floor{ \frac{y}{p^{\frac{1-k}{2-k}}}}\bigr) \rightarrow \exp \bigl(-   & xy - 2\alpha (1 -k) x^k y^2 - 2 \alpha^2 (1-k)^2 x^{2k-1} y^3 \\
& - \alpha^3 (1 -k)^3 x^{3k-2} y^4 - \frac{\alpha^4 (1-k)^4 x^{4k-3} y^5}{5} \bigr),
\end{split}
\end{equation}
as $p \rightarrow 0$.
This proves the convergence of $p^{\frac{1-k}{2-k}} G_{\floor{\frac{x}{p^{\frac{1}{2-k}}}}}^p$ to $\overline{G}_x$. 

We now show uniform convergence of $\mathbb{P}\Bigl(   p^{\frac{1-k}{2-k}}  G_{\floor{\frac{x}{p^{\frac{1}{2-k}}}}}^p \geq  y\Bigr)$ to $\mathbb{P}(\overline{G}_x \geq y)$. We assume that $x$, $y$ are bounded above by $M$. Taking logarithms on both sides of \eqref{eqn:Gbarx_CTCP4}, we get
\begin{equation}
\label{eqn:Gbarx_CTCP10}
\begin{split}
\log \mathbb{P}\Bigl( & G_{\floor{\frac{x}{p^{\frac{1}{2-k}}}}}^p  > \floor{ \frac{y}{p^{\frac{1-k}{2-k}}}}\Bigr) \\
%\log  \mathbb{P}\Biggl( p^{\frac{1-k}{2-k}} & G_{\floor{\frac{x}{p^{\frac{1}{2-k}}}}}^p \geq  y\Biggr)    \\ 
%\in &  \Biggl[ \Bigl(-px_0(1 + (1 + z)^{\frac{1}{1-k}} + \cdots + (1 + mz)^{\frac{1}{1-k}} - p(m+1) \Bigr), \\
% & \hspace*{0.2cm}  \Bigl(-px_0(1 + (1 + z)^{\frac{1}{1-k}} + \cdots + (1 + mz)^{\frac{1}{1-k}}) + p(m+1) \Bigr) \Biggr].
\in \Bigl[& (x_0(1 + (1 + z)^{\frac{1}{1-k}} + \cdots + (1 + mz)^{\frac{1}{1-k}}) + (m+1)) \log(1-p), \\
 & (x_0(1 + (1 + z)^{\frac{1}{1-k}} + \cdots + (1 + mz)^{\frac{1}{1-k}}) - (m+1)) \log(1-p) \Bigr],
\end{split}
\end{equation}

The equation \eqref{eqn:Gbarx_CTCP10} has elements of the form $x_0 m^n z^j \log(1-p)$ with $n - 1 \leq j < 5$, $0 < n \leq 5$ and the terms $(m+1)\log(1-p)$ and $x_0 \log(1-p)$. The elements with $j = (n-1)$ are the only terms that contribute to the limit, i.e., equation \eqref{eqn:Gbarx_CTCP9}. From \eqref{eqn:Gbarx_CTCP8}, the remaining terms that do not contribute to the limit are of the form $c(n,j) x^{1 - j + kj} y^n p^{\frac{-1+(j-n)(1-k)}{2-k}} \log(1-p)$ with $j > (n-1)$ and $c(n,j)$ being some finite coefficient. Additionally, we have terms of the form $x_0 \log(1-p)$ and $(m+1)\log(1-p)$, which also do not contribute to \eqref{eqn:Gbarx_CTCP9}. For $n -1 < j$, $p^{\frac{-1+(j-n)(1-k)}{2-k}}$ is of the form $p^{\epsilon(n,j,k)}$ with $\epsilon(n,j,k) > -1$. The terms can be grouped together and written as $f(x,y, p)$, where $f$ has elements of the form $x^{1 - j + kj} y^n p^{\epsilon(n,j,k)} \log(1-p)$, with $p^{\epsilon(n,j,k)} > -1$.  Let $T = \max \{1, M\}$, hence $x^{1 - j + kj} y^n \leq T^{1 - j + kj + n}$, for $x,y \leq M$. Therefore, for $x,y \leq M$, we have,
\begin{equation}
\label{eqn:Gbarx_CTCP11}
\begin{split}
\Bigl| \log & \mathbb{P}\Bigl(  G_{\floor{\frac{x}{p^{\frac{1}{2-k}}}}}^p > \floor{ \frac{y}{p^{\frac{1-k}{2-k}}}} \Bigr)    - \log \mathbb{P}(\overline{G}_x \geq y) \Bigr|  \\
& \leq |f(x,y, p)| + x_0 \log(1-p) + (m+1) \log(1-p), \\ 
& \leq c_1 T^{5 - \frac{j}{4}} p^\epsilon \log(1-p) + T p^{-\frac{4}{5}} \log(1-p) + (T p^{-\frac{1}{5}} + 1) \log(1-p) , \\
& \leq c_1  T^5 p^\epsilon \log(1-p) + T p^{-\frac{4}{5}} \log(1-p) + (T p^{-\frac{1}{5}} + 1) \log(1-p), 
\end{split}
\end{equation}
where the term $T^{5 - \frac{j}{4}}$ in the inequality comes from the element with the largest power for $x$ and $y$ in $f(x,y, p)$, $\epsilon = \displaystyle{\min_{(n,j): j > (n -1), n \leq 4}} \epsilon(n,j,k) > 0$ and $c_1$ is a constant independent of $p, x, y$. Therefore, we have
\begin{equation}
\label{eqn:Gbarx_CTCP12}
\lim_{p \rightarrow 0} \sup_{p^{\frac{1}{2-k}} \leq x \leq M, y \leq M } \Bigl| \log  \mathbb{P}\Bigl(   G_{\floor{\frac{x}{p^{\frac{1}{2-k}}}}}^p > \floor{ \frac{y}{p^{\frac{1-k}{2-k}}}}  \Bigr)   - \log \mathbb{P}(\overline{G}_x \geq y) \Bigr|  = 0.
\end{equation}
We can similarly prove
\begin{equation}
\label{eqn:Gbarx_CTCP12A}
\lim_{p \rightarrow 0} \sup_{p^{\frac{1}{2-k}} \leq x \leq M, y \leq M } \Bigl| \log  \mathbb{P}\Bigl(   G_{\floor{\frac{x}{p^{\frac{1}{2-k}}}}}^p \geq \floor{ \frac{y}{p^{\frac{1-k}{2-k}}}}  \Bigr)   - \log \mathbb{P}(\overline{G}_x \geq y) \Bigr|  = 0.
\end{equation}
The result \eqref{eqn:Gbarx_CTCP3} in Proposition \ref{prop:Gbarx_CTCP1} then follows from the uniform continuity of the $\exp()$ function on $(-\infty, 0).$

To prove uniform boundedness of the moment generating functions of $p^{\frac{1-k}{2-k}}$  $G_{\floor{\frac{x}{p^{\frac{1}{2-k}}}}}^p$, we require a uniform bound on their tail distributions. For a fixed $p$, $ G_{\floor{\frac{x}{p^{\frac{1}{2-k}}}}}^p$ is stochastically dominated by $G_{1}^p$, for all $x > p^{\frac{1}{2-k}}$. Let $z_0 = \alpha(1-k)$ and $m = \floor{\frac{y}{p^{\frac{1-k}{2-k}}}}$. Therefore, for all $x > p^{\frac{1}{2-k}}$, we have
\begin{equation}
\label{eqn:Gbarx_CTCP_moment0}
\begin{split}
\mathbb{P}\Bigl( G_{\floor{\frac{x}{p^{\frac{1}{2-k}}}}}^p  \geq \floor{ \frac{y}{p^{\frac{1-k}{2-k}}}}\Bigr) 
\leq   (1-p)^{(1 + (1 + z_0)^{\frac{1}{1-k}} + \cdots + (1 + (m-1)z_0)^{\frac{1}{1-k}} - m)},
%& \exp(-p(1 + (1 + z_0)^{\frac{1}{1-k}} + \cdots + (1 + (m-1)z_0)^{\frac{1}{1-k}} - m).
\end{split}
\end{equation}
Since, $(1-p) \leq \exp(-p)$, we have,
\begin{equation}
\label{eqn:Gbarx_CTCP_moment0A}
\begin{split}
\mathbb{P}\Bigl( G_{\floor{\frac{x}{p^{\frac{1}{2-k}}}}}^p  \geq \floor{ \frac{y}{p^{\frac{1-k}{2-k}}}}\Bigr) 
\leq  e^{(-p(1 + (1 + z_0)^{\frac{1}{1-k}} + \cdots + (1 + (m-1)z_0)^{\frac{1}{1-k}} - m)}.
\end{split}
\end{equation}
Taking $\log$ and using \eqref{eqn:Gbarx_CTCP1_start}, we get
\begin{equation}
\label{eqn:Gbarx_CTCP_moment}
\begin{split}
\log \mathbb{P}\Bigl(  & p^{\frac{1-k}{2-k}}  G_{\floor{\frac{x}{p^{\frac{1}{2-k}}}}}^p \geq  y\Bigr)    \\ 
& \leq  -p (1 + (1 + z_0)^{\frac{1}{1-k}} + \cdots + (1 + (m-1)z_0)^{\frac{1}{1-k}}) + pm , \\
& = -p (4z_0 \sum_{i=1}^{m-1} i + 6z_0^2 \sum_{i=1}^{m-1} i^2 + 4z_0^3 \sum_{i=1}^{m-1} i^3 + z_0^4 \sum_{i=1}^{m-1} i^4 ), \text{ (from \eqref{eqn:Gbarx_CTCP7})}\\
& \leq -p z_0^4 \sum_{i=1}^{m} i^4 , \\
& = -p z_0^4 \Bigl(\frac{(m-1)^5}{5} +  \frac{(m-1)^4}{2} + \frac{(m-1)^3}{3} - \frac{m-1}{30}\Bigr), \\
& \leq -p z_0^4 \Bigl( \frac{(m-1)^5}{5} - \frac{m-1}{30} \Bigr), \\
& \leq - \frac{pz_0^4 (m-1)^5}{5}   + p^{\frac{4}{5}} \frac{z_0^4y}{30}, \hspace{0.8cm}\text{(since $m = \floor{\frac{y}{p^{\frac{1-k}{2-k}}}} \leq  \frac{y}{p^{\frac{1-k}{2-k}}} + 1$)} \\
& \leq - \frac{p z_0^4 (\frac{m}{2})^5}{5}   + \frac{z_0^4y}{30}, \\
& \leq - \frac{z_0^4 y^5}{160}   + \frac{z_0^4y}{30}, 
\end{split}
\end{equation}
Therefore, for all $x > p^{\frac{1}{2-k}}$ and for all $p$ the tail distribution for $p^{\frac{1-k}{2-k}} G_{\floor{\frac{x}{p^{\frac{1}{2-k}}}}}^p $ is upper bounded by $\exp ( -  \alpha^4 (1 -k)^4 (\frac{ y^5}{160}   - \frac{y}{30}))$. We can find constants $c_1 > 0$, $c_2 > 0$, such that 
\begin{equation}
\mathbb{P}\Bigl(p^{\frac{1-k}{2-k}}  G_{\floor{\frac{x}{p^{\frac{1}{2-k}}}}}^p \geq  y\Bigr) \leq c_1 \exp(-c_2 y),
\end{equation}
for all $p \in (0,1)$, for all $x > p^{\frac{1}{2-k}}$. We get the desired result using Proposition \ref{prop:mgf_bound} in Appendix \ref{app:appendixA}. 
\end{proof}

We now derive a limiting result for the Markov chain, $\{V_k(p)\}$ embedded at the loss epochs of the TCP Compound window evolution process. Suppose $V_0(p) = x$, then  $V_1(p)$ is given by
\begin{equation}
\label{eqn:Vn_of_p}
V_1(p) = (1 - \beta) \Bigl(x^{1-k} + \alpha (1 - k) G_x^p\Bigr)^\frac{1}{1-k},
\end{equation}
where $G_x^p$ denotes time (in multiples of $R$) for first loss given that the initial window size is $x$.

We now define a Markov chain, $\{\overline{V}_n\}$, which serves as the limit for the process $\{V_n(p)\}$ with appropriate scaling. Let $\overline{V}_0$ be a random variable with an arbitrary initial distribution on $\mathbb{R}^+$. Define $\overline{V}_n$ for $n \geq 1$ as
\begin{equation}
\label{eqn:lim_Vn_def}
\overline{V}_n = (1 - \beta) \Bigl(\overline{V}_{n-1}^{1-k} + \alpha (1 - k) (\overline{G}_{\overline{V}_{n-1}}) \Bigr)^\frac{1}{1-k},
\end{equation}
where $\{\overline{G}_{\overline{V}_{n-1}}\}$ are random variables with distribution given by \eqref{eqn:Gbarx_CTCP2}. Also,  $\overline{G}_{\overline{V}_{n-1}}$ is chosen independently of $\{\overline{V}_k: k < n-1\}$. 
\begin{proposition}
\label{prop:Vbarx_CTCP}
Suppose $\overline{V}_0 = x$ and $V_0(p) = \floor{\frac{x}{p^{\frac{1}{2-k}}}}$ for some $x > 0$, for all $p>0$. Then we have 
\begin{equation}
\label{eqn:Vbarx_CTCP2}
\begin{split}
\lim_{p \rightarrow 0} \sup_{x \geq p^{\frac{1}{2-k}} }  \Bigl| & \mathbb{P}_x(p^{\frac{1}{2 -k}}V_1(p) \leq a_1, p^{\frac{1}{2 -k}}V_2(p) \leq a_2, \cdots, p^{\frac{1}{2 - k}}V_n(p) \leq a_n  ) \\
& - \mathbb{P}_x(\overline{V}_1 \leq a_1, \overline{V}_2 \leq a_2, \cdots, \overline{V}_n \leq a_n)  \Bigr| = 0,
\end{split}
\end{equation}
where $ a_i \in \mathbb{R}^+$, for $i = 1, 2, \cdots, n$ and $\mathbb{P}_x$ denotes the law of the processes when $\overline{V}_0 = x$ and $V_0(p) = \floor{\frac{x}{p^{\frac{1}{2-k}}}}$.
\end{proposition}
\begin{proof}
We prove \eqref{eqn:Vbarx_CTCP2} for $n=1,2$, the proof for $n > 2$ follows by induction. The proof for $n=1$ is given below.
\begin{equation}
\label{eqn:Vbarx_CTCP3}
\begin{split}
\mathbb{P}_x( p^{\frac{1}{2 -k}}V_1(p) \leq a_1) &= \mathbb{P} \Biggl( (1 -\beta)x \Bigl(1 + \frac{\alpha(1-k)G^{p}_{\floor{\frac{x}{p^{\frac{1}{2-k}}}}} }{ \Bigl(\frac{x}{p^{\frac{1}{2-k}}}\Bigr)^{1-k} } \Bigr)^{\frac{1}{1-k}} \leq a_1   \Biggr), \\
%&= \mathbb{P} \Bigl( (1 -\beta)x \Biggl(1 + \frac{\alpha(1-k) p^{\frac{1-k}{2-k}} G^{p}_{\floor{\frac{x}{p^{\frac{1}{2-k}}}}} }{x^{1-k}} \Biggr)^{\frac{1}{1-k}} \leq a_1  \Bigr), \\
&= \mathbb{P} \Biggl( p^{\frac{1-k}{2-k}} G^{p}_{\floor{\frac{x}{p^{\frac{1}{2-k}}}}}  \leq \frac{\Bigl(\frac{a_1}{1 -\beta}\Bigr)^{1-k} - x^{1-k}}{\alpha (1 -k)} \Biggr).
\end{split}
\end{equation}
Similarly, 
\begin{equation}
\label{eqn:Vbarx_CTCP3A}
\mathbb{P}_x( \overline{V}_1(p) \leq a_1) = \mathbb{P} \Biggl( \overline{G}_{x}  \leq \frac{\Bigl(\frac{a_1}{1 -\beta}\Bigr)^{1-k} - x^{1-k}}{\alpha (1 -k)} \Biggr).
\end{equation}
From equation \eqref{eqn:Gbarx_CTCP4} in Proposition \ref{prop:Gbarx_CTCP1},  $\mathbb{P} \Bigl( p^{\frac{1-k}{2-k}} G^{p}_{\floor{\frac{x}{p^{\frac{1}{2-k}}}}}$ $\leq y \Bigr)$ converges to $\mathbb{P} ( \overline{G}_x \leq y)$  uniformly in $x$ and $y$ over any bounded interval. Also, from \eqref{eqn:Vn_of_p} and \eqref{eqn:lim_Vn_def}, for $x > \frac{a_1}{1 - \beta}$, we have $\mathbb{P}_x( p^{\frac{1}{2 -k}}V_1(p) \leq a_1) $ $=$ $\mathbb{P}_x( \overline{V}_1 \leq a_1) $ $= 0$. Therefore, we have
\begin{equation}
\label{eqn:Vbarx_CTCP4}
\lim_{p \rightarrow 0} \sup_{x \geq p^{\frac{1}{2-k}} } \Bigl| \mathbb{P}_x \Bigl( p^{\frac{1}{2 -k}}V_1(p) \leq a_1   \Bigr) - \mathbb{P}_x \Bigl( \overline{V}_1 \leq a_1  \Bigr) \Bigr| = 0.
\end{equation}
This proves \eqref{eqn:Vbarx_CTCP2} for $n=1$. 
We use Proposition \ref{prop:Gbarx_CTCP1} for proving the result for $n=2$. Consider
\begin{equation}
\label{eqn:Vbarx_CTCP6}
\begin{split}
\mathbb{P}_x(& p^{\frac{1}{2 -k}}V_1(p) \leq a_1, p^{\frac{1}{2 -k}}V_2(p) \leq a_2) \\
= \displaystyle\int\limits_{0}^{a_1} & \mathbb{P} \Bigl( (1 -\beta)y \Biggl(1 + \frac{\alpha(1-k)G^{p}_{\floor{\frac{y}{p^{\frac{1}{2-k}}}}} }{ \Bigl(\frac{y}{p^{\frac{1}{2-k}}}\Bigr)^{1-k} } \Biggr)^{\frac{1}{1-k}} \leq a_2   \Bigr) \mathbb{P}_x(p^{\frac{1}{2 -k}}V_1(p)  \in dy) \\
= \displaystyle\int\limits_{0}^{a_1} & \mathbb{P} \Bigl( (1 -\beta)y \Biggl(1 + \frac{\alpha(1-k) p^{\frac{1-k}{2-k}} G^{p}_{\floor{\frac{y}{p^{\frac{1}{2-k}}}}} }{ y^{1-k} } \Biggr)^{\frac{1}{1-k}} \hspace*{-2mm} \leq a_2   \Bigr)  \mathbb{P}_x(p^{\frac{1}{2 -k}}V_1(p)  \in dy).
\end{split}
\end{equation}
From equation \eqref{eqn:Gbarx_CTCP4} in Proposition \ref{prop:Gbarx_CTCP1},  $\mathbb{P} \Bigl( p^{\frac{1-k}{2-k}} G^{p}_{\floor{\frac{x}{p^{\frac{1}{2-k}}}}}$ $\leq y \Bigr)$ converges to $\mathbb{P} ( \overline{G}_x \leq y)$  uniformly in $x$ and $y$ over any bounded interval. Therefore, for any given $\epsilon >0$ there exists a $p^*$ such that for  $p < p^*$, we have  
\begin{equation}
\label{eqn:Vbarx_CTCP7}
\begin{split}
\mathbb{P}_x(& p^{\frac{1}{2 -k}}V_1(p) \leq a_1, p^{\frac{1}{2 -k}}V_2(p) \leq a_2) \\
& \approx_\epsilon  \displaystyle\int\limits_{0}^{a_1}  \mathbb{P} \Bigl( (1 -\beta)y \Bigl(1 + \frac{\alpha(1-k) \overline{G}_{y}}{ y^{1-k} } \Bigr)^{\frac{1}{1-k}} \leq a_2  \Bigr) \mathbb{P}_x(p^{\frac{1}{2 -k}}V_1(p)  \in dy), \\
& =  \displaystyle\int\limits_{0}^{\infty}  \mathbb{P} \Bigl( (1 -\beta)y \Bigl(1 + \frac{\alpha(1-k) \overline{G}_{y}}{ y^{1-k} } \Bigr)^{\frac{1}{1-k}} \leq a_2  \Bigr) \mathds{1}_{\{y \leq a_1\}} \mathbb{P}_x(p^{\frac{1}{2 -k}}V_1(p)  \in dy), \\
& = \mathbb{E}_x [ g(p^{\frac{1}{2-k}} V_1(p))],
\end{split}
\end{equation}
where the symbol $\approx_\epsilon$ denotes that the RHS of the expression is $\epsilon$-close to the LHS and the function $g(y) = \mathbb{P} \Bigl( (1 -\beta)y \Bigl(1 + \frac{\alpha(1-k) \overline{G}_{y}}{ y^{1-k} } \Bigr)^{\frac{1}{1-k}} \leq a_2  \Bigr) \mathds{1}_{\{y \leq a_1\}}$. For any continuous functions $f$ on $\mathbb{R}^+$ with compact support, using Proposition \ref{prop:unif_conv} from appendix \ref{app:appendixA}, we have
\begin{equation}
\label{eqn:Vbarx_CTCP5}
\lim_{p \rightarrow 0} \sup_{x \geq p^{\frac{1}{2-k}} } \Bigl| E_x \Bigl[ f(p^{\frac{1}{2 -k}}V_1(p)) \Bigr] - E_x \Bigl[ f(\overline{V}_1) \Bigr] \Bigr| = 0,
\end{equation}
The function $g$ is continuous with a compact support, therefore using \eqref{eqn:Vbarx_CTCP5} we get,
\begin{equation}
\label{eqn:Vbarx_CTCP8}
\lim_{p \rightarrow 0} \sup_{x \geq p^{\frac{1}{2-k}} } \Bigl| \mathbb{P}_x \Bigl( p^{\frac{1}{2 -k}}V_1(p) \leq a_1, p^{\frac{1}{2 -k}}V_2(p) \leq a_2   \Bigr) - \mathbb{P}_x \Bigl( \overline{V}_1 \leq a_1, \overline{V}_2 \leq a_2   \Bigr) \Bigr| = 0.
\end{equation}
The proof of \eqref{eqn:Vbarx_CTCP2} for $n>2$ can be done using induction.
\end{proof}
Since the finite dimensional distributions of $ \{ p^{\frac{1}{2-k}}V_n(p) \} $ converge to $ \{ \overline{V}_n \}$, we have
\begin{corollary}
\label{coro:Vbarx_CTCP}
If $\lim_{p \rightarrow 0}  p^{\frac{1}{2-k}}V_0(p)$ converges in distribution to $\overline{V}_0$, then the Markov chain $ \{ p^{\frac{1}{2-k}}V_n(p) \} $ converges in distribution to the Markov chain $ \{ \overline{V}_n \}$.
\end{corollary}

The limiting distribution of the Markov chain $ \{ \overline{V}_n \}$ (if it exists) can be used to compute approximation to the time average throughput and average window size of TCP Compound for small enough $p$. The following proposition shows the existence of the limiting distribution for $\{ \overline{V}_n \} $ process.
\begin{proposition}
\label{prop:Vbarx_invariance}
The Markov chain $\{ \overline{V}_n \} $ is Harris recurrent and has a unique invariant distribution.
\end{proposition}
\begin{proof}
For the Markov chain $\{ \overline{V}_n \}$, from \eqref{eqn:lim_Vn_def}, we have
\begin{equation}
\overline{V}_n^{1-k} = (1 - \beta)^{1-k} \Bigl(\overline{V}_{n-1}^{1-k} + \alpha (1 - k) (\overline{G}_{\overline{V}_{n-1}})_{n-1} \Bigr).
\end{equation}
Consider the sequence $\{X_n\}$ defined as follows. Let $X_0$ have the same distribution as $\overline{V}_0$. Choose a sequence of i.i.d random variables $\{\eta_{n}\}$ with same distribution as of $\overline{G}_0$ and set $X_n$ for $n \geq 1$ recursively as follows
\begin{equation}
X_n = (1 - \beta)^{1-k} \Bigl(X_{n-1} + \alpha (1 - k) \eta_{n-1} \Bigr).
\end{equation}
The term $(1 - \beta)^{1-k}$ is strictly less than $1$ and $\eta_n$ is absolutely continuous with respect to the Lebesgue measure. Therefore (see example $14.2.15$ in \cite{Athreya2006}) the Markov chain $\{X_n\}$ is Harris recurrent and admits a unique stationary distribution, $\pi(\cdot)$, and $\mathbb{P}_x(X_n \in \cdot)$ converges to $\pi(\cdot)$ for all $x$ in total variation.  Since $X_n$ converges to a stationary distribution, it is tight. From \eqref{eqn:Gbarx_CTCP2}, for any $x > 0$, $\overline{G}_x$  is stochastically smaller than $\overline{G}_0$. Therefore for each $n$, $\overline{V}_n^{1-k}$ is stochastically smaller than $X_n$ and consequently the sequence of distributions of $\sum_{i=1}^{n} \frac{1}{n} \mathbb{P}_x(\overline{V}_i^{1-k} \in \cdot)$ is tight. Therefore, the sequence of distributions of $\sum_{i=1}^{n} \frac{1}{n} \mathbb{P}_x(\overline{V}_i \in \cdot)$ is also tight. The limit of $ \sum_{i=1}^{n} \frac{1}{n}\mathbb{P}_x(\overline{V}_i \in \cdot)$ is an invariant probability distribution for the Markov chain, $\{\overline{V}_n\}$.
\end{proof}

In Proposition \ref{prop:Vbarx_invariance}, we have shown that the Markov chain $\{\overline{V}_n\}$ has a unique invariant distribution. Also for $p > 0$,  $p^{\frac{1}{2-k}}V_n(p)$  has a unique invariant distribution (see Proposition \ref{prop:ergodicity} for existence of unique invariant distribution for $W_n(p)$, existence of unique invariant distribution for $p^{\frac{1}{2-k}}V_n(p)$ follows from Corollary \ref{coro:ergodicity}). In Proposition \ref{prop:stat_dist_conv}, we see the relation between these invariant distributions. The proof for Proposition \ref{prop:stat_dist_conv} requires Proposition \ref{prop:convergence_hitting_times} and Lemma \ref{lemma:hitting_time}.

Let $\pi_p$ denote the invariant distribution of $p^{\frac{1-k}{2-k}} V_n(p)^{1-k}$ and let $\overline{\pi}$ denote the invariant distribution of $\overline{V}_n^{1-k}$. For $K > 0$, define $T_K(p)$ to be the hitting time to set $[0,K]$ by the process $p^{\frac{1-k}{2-k}}V_n(p)^{1-k}$ to hit $[0,K]$, similarly $\overline{T}_K$ is the corresponding hitting time for the $\overline{V}_n^{1-k}$ process. Let $\pi_p^K(\cdot) = \frac{\pi_p \mathds{1}_{K}(\cdot) }{\int_K \pi_p(x) dx}$, similarly let $\overline{\pi}^K(\cdot) = \frac{\overline{\pi}\mathds{1}_{K}(\cdot)}{ \int_K \overline{\pi}(x) dx }$. 
\begin{proposition}
\label{prop:convergence_hitting_times}
For any $n$ $\in \{1,2,3, \cdots \}$,
\begin{equation}
\lim_{p \rightarrow 0} \mathbb{P}_{\pi_p^K}(T_K(p) = n) = \mathbb{P}_{\overline{\pi}^K}(\overline{T}_K = n)
\end{equation}
\end{proposition}
\begin{proof}
Let $\mathbb{P}_x$ denote the law of the processes, $\{p^{\frac{1-k}{2-k}} V_n(p)^{1-k}\}$ and $\{\overline{V}_n^{1-k}\}$ when $p^{\frac{1-k}{2-k}} V_0(p)^{1-k} = x$ and $\overline{V}_0^{1-k} = x$. Suppose that $p^{\frac{1-k}{2-k}} V_0(p)^{1-k} = x \leq K$, then for the hitting time $T_K(p)$, we have
\begin{equation}
\label{eqn:stat_dist_conv_Tkp}
\begin{split}
\mathbb{P}_x(T_K(p) = n) = \mathbb{P}_x(& p^{\frac{1-k}{2-k}} V_1(p)^{1-k}  > K,  p^{\frac{1-k}{2-k}} V_2(p)^{1-k} > K,  \cdots, \\ 
& p^{\frac{1-k}{2-k}} V_{n-1}(p)^{1-k} > K, p^{\frac{1-k}{2-k}} V_n(p)^{1-k} \leq K).
\end{split}
\end{equation}
Similarly, assuming $\overline{V}_0^{1-k} = x \leq K$, we have
\begin{equation}
\label{eqn:stat_dist_conv_Tkbar}
\mathbb{P}_x(\overline{T}_K = n) = \mathbb{P}_x(\overline{V}_1^{1-k} > K, \overline{V}_2^{1-k} > K, \cdots, \overline{V}_{n-1}^{1-k} > K,  \overline{V}_n^{1-k} \leq K).
\end{equation}
From Proposition \ref{prop:Vbarx_invariance}, for any $x$,
%as $p \rightarrow 0$, $T_K(p)$ converges in distribution to $\overline{T}_K$, if $p^{\frac{1-k}{2-k}} V_0(p)^{1-k} = \overline{V}_0^{1-k} = x \leq K$. 
\begin{equation}
\lim_{p \rightarrow 0} \mathbb{P}_x(T_K(p) = n) = \mathbb{P}_x(\overline{T}_K = n).
\end{equation}
The sequence $\pi_p^K$ is tight and hence there exists a subsequence $\{p_n\}$ with $p_n \rightarrow 0$ as $n \rightarrow \infty$ and distribution $\pi^K$ such that $\pi_{p_n}^K$ converges to $\pi^K$ as $n \rightarrow \infty$. From Corollary \ref{coro:Vbarx_CTCP}, along this sequence, 
\begin{equation}
\lim_{m \rightarrow \infty} \mathbb{P}_x(T_K(p_m) = n) = \mathbb{P}_{\pi^K}(\overline{T}_K = n).
\end{equation}
Let us define sequences, $\{Y_n(p)\}$ , $\{\overline{Y}_n\}$ which denote the sequence of states visited in successive visits to $[0,K]$ by the Markov chains, $\{p^{\frac{1-k}{2-k}} V_n(p)^{1-k} \}$ and $\{\overline{V}_n^{1-k} \}$ respectively. Since, $\pi_p$ is the stationary distribution for $\{p^{\frac{1-k}{2-k}} V_n(p)^{1-k}\}$, we have
\begin{equation}
\mathbb{P}_{\pi_{p_n}^K}(Y_0(p_n) \leq a) = \mathbb{P}_{\pi_{p_n}^K}(Y_1(p_n) \leq a).
\end{equation}
As $n \rightarrow \infty$, $\pi_{p_n}^K$ converges to $\pi^K$. From Corollary \ref{coro:Vbarx_CTCP}, along this sequence, we have
\begin{equation}
\pi^K([0,a])  \doteqdot \mathbb{P}_{\pi^K}(\overline{Y}_0 \leq a) = \lim_{n \rightarrow \infty} \mathbb{P}_{\pi_{p_n}^K}(Y_1(p_n) \leq a)  = \mathbb{P}_{\pi^K}(\overline{Y}_1 \leq a).
\end{equation}
Thus, $\pi^K$ is invariant distribution for the embedded Markov chain, $\{\overline{Y}_n\}$. However, since $\overline{\pi}$ is the invariant distribution of  $\overline{V}_n$, $\overline{\pi}_K$ is the invariant distribution of the embedded Markov chain, $\{\overline{Y}_n\}$. Therefore, by uniqueness of invariant distribution, we get $\pi^K = \overline{\pi}^K$. The above argument holds for any convergent subsequence of any arbitrarily chosen subsequence $\{p_n\}$ going to $0$ as $n \rightarrow \infty$. Therefore, we have
\begin{equation}
\mathbb{P}_{\pi_p^K}(T_K(p) = n) \rightarrow \mathbb{P}_{\overline{\pi}^K}(\overline{T}_K = n),
\end{equation}
as $p \rightarrow 0$.
\end{proof}

\begin{lemma}
\label{lemma:hitting_time}
There exist constants $K$, $\zeta$, and $\lambda > 0$ such that for all $0 < p < 1$, we have
\begin{equation}
\label{eqn:hitting_time1}
\mathbb{E}\Bigl[ e^{\zeta T_K(p) - \lambda p^{\frac{1-k}{2-k}}V_0(p)^{1-k}} \Big| p^{\frac{1-k}{2-k}}V_0(p)^{1-k} > K \Bigr] \leq 1.
\end{equation}
\end{lemma}
\begin{proof}
Let $\mathcal{F}_n$ be the sigma-algebra generated by $V_0(p), V_1(p), \cdots, V_n(p)$. Define $Z_n = e^{\lambda p^{\frac{1-k}{2-k}}V_n^{1-k}}$, for $n \in \mathbb{N}$. Using \eqref{eqn:Vn_of_p}, we have
\begin{equation}
\begin{split}
Z_{n+1} &= e^{ \lambda p^{\frac{1-k}{2-k}} (1 - \beta)^{1-k} \Bigl( V_n^{1-k} + \alpha (1 -k) G_{V_n(p)}^p \Bigr)  } \\
&= Z_n e^{\lambda p^{\frac{1-k}{2-k}} ((1 - \beta)^{1-k} - 1)  V_n^{1-k}} e^{ \lambda p^{\frac{1-k}{2-k}} \alpha (1 -k) (1 - \beta)^{1-k} G_{V_n(p)}(p)}
\end{split}
\end{equation}
Therefore, on the event $E_n = \{ p^{\frac{1-k}{2-k}}V_n(p)^{1-k} > K \}$, the drift in the $\{Z_n\}$ process can be bounded as follows,
\begin{equation}
\begin{split}
\mathbb{E}[Z_{n+1} - Z_n | \mathcal{F}_n] & \leq Z_n \Bigl( e^{- \lambda_0 K} \mathbb{E}\Bigl[e^{ \lambda_1 p^{\frac{1-k}{2-k}} G_{V_n(p)}} \Big| E_n\Bigr] - 1 \Bigr), \\
& \leq Z_n  \Bigl( e^{- \lambda_0 K} \sup_{p \in (0,1); y \geq 0} \mathbb{E}\Bigl[\exp \bigl( \lambda_1 p^{\frac{1-k}{2-k}} G_{\floor{\frac{y}{p^{\frac{1}{2-k}}}}}\bigr)\Bigr] - 1 \Bigr), \\
\end{split}
\end{equation}
where $\lambda_0 = \lambda ((1 - \beta)^{1-k} - 1) $, $\lambda_1 = \lambda \alpha (1 -k) (1 - \beta)^{1-k} $.
From Proposition \ref{prop:Gbarx_CTCP1}, there exists a $\lambda_2 > 0$ such that for $0 < \lambda < \lambda_2$, there exists a finite $C$ such that
\begin{equation}
\mathbb{E}[Z_{n+1} - Z_n | \mathcal{F}_n]  \leq Z_n \Bigl(C e^{- \lambda_0 K}  - 1 \Bigr),
\end{equation}
on the event $E_n = \{ p^{\frac{1-k}{2-k}}V_n(p)^{1-k} > K \}$. We can now choose a $K_0$ such that for all $K \geq K_0$, there exists a $\eta < 1$ such that,
\begin{equation}
\mathbb{E}[Z_{n+1} | \mathcal{F}_n]  \leq \eta Z_n,
\end{equation}
on the event $E_n = \{ p^{\frac{1-k}{2-k}}V_n(p)^{1-k} > K \}$. Therefore if $p^{\frac{1-k}{2-k}}V_0(p)^{1-k} > K  $, the sequence $U_n = \eta^{- n \wedge T_K(p)} Z_{n \wedge T_K(p)}$ is a super-martingale. Hence for $n \geq 0$, we have
\begin{equation}
\mathbb{E} [\eta^{- n \wedge T_K(p)} Z_{n \wedge T_K(p)} | \mathcal{F}_0] = \mathbb{E}[U_n | \mathcal{F}_0] \leq U_0 = e^{\lambda p^{\frac{1-k}{2-k}}V_0(p)^{1-k}}.
\end{equation}
Since, $Z_n \geq 1$, we have
\begin{equation}
\mathbb{E} [\eta^{- n \wedge T_K(p)} ] \leq  e^{\lambda p^{\frac{1-k}{2-k}}V_0(p)^{1-k}},
\end{equation}
on the event $\{p^{\frac{1-k}{2-k}}V_0(p)^{1-k} > K  \}$. Now using Fatou's lemma, we get
\begin{equation}
\mathbb{E} [\eta^{T_K(p)} ] \leq  e^{\lambda p^{\frac{1-k}{2-k}}V_0(p)^{1-k}},
\end{equation}
on the event $\{p^{\frac{1-k}{2-k}}V_0(p)^{1-k} > K  \}$ which proves the desired result.
\end{proof}

\begin{proposition}
\label{prop:stat_dist_conv}
The invariant distribution of $p^{\frac{1}{2-k}}V_n(p)$ converges weakly to the invariant distribution of $\overline{V}_n$ as $p \rightarrow 0$.
\end{proposition}
\begin{proof}
We will first prove that the invariant distribution of $p^{\frac{1-k}{2-k}} V_n(p)^{1-k}$ converges to the invariant distribution of $\overline{V}_n^{1-k}$. The continuous mapping theorem then gives us the desired result. Suppose $p^{\frac{1-k}{2-k}} V_0(p)^{1-k} \in [0,K]$ and set its initial distribution as $\pi_p^K(\cdot) = \frac{\pi_p \mathds{1}_{K}(\cdot) }{\int_K \pi_p(x) dx}$, similarly let $\overline{V}_0^{1-k} \in [0,K]$ with initial distribution $\overline{\pi}^K(\cdot) = \frac{\overline{\pi}\mathds{1}_{K}(\cdot)}{ \int_K \overline{\pi}(x) dx } $. Using Palm calculus\cite{Asmussen}, for any bounded measurable function $f$ we have
\begin{equation}
\mathbb{E}_{\pi_p}[f] \defeq \int f d\pi_p = \frac{1}{\mathbb{E}_{\pi_p^K}[T_{K}(p)]} \mathbb{E}_{\pi_p^K} \Bigl[  \displaystyle{\sum\limits_{n =0}^{T_{K}(p) - 1}} f(p^{\frac{1-k}{2-k}} V_n(p)^{1-k} ) \Bigr],
\end{equation}
and
\begin{equation}
\mathbb{E}_{\overline{\pi}}[f] \defeq \int f d\overline{\pi} = \frac{1}{\mathbb{E}_{\overline{\pi}^K}[\overline{T}_{K}]} \mathbb{E}_{\overline{\pi}^K} \Bigl[  \displaystyle{\sum\limits_{n =0}^{\overline{T}_{K} - 1}} f(\overline{V}_n^{1-k} ) \Bigr].
\end{equation}

Now, $T_K(p) = 1$ if $ p^{\frac{1-k}{2-k}} V_1(p)^{1-k} \leq K $ and using Markov property, $T_K(p) = 1 + T'_K(p)$ if $ p^{\frac{1-k}{2-k}} V_1(p)^{1-k} > K $, where $T'_K(p)$ denotes hitting time to $[0,K]$ for Markov chain starting from $p^{\frac{1-k}{2-k}} V_1(p)^{1-k}$. Therefore we have,  
\begin{equation}
\begin{split}
\mathbb{E}_{\pi_p^K} \Bigl[ e^{\zeta T_K(p)} \Bigr] & \leq e^{\zeta} +  e^{\zeta} \mathbb{E}_{\pi_p^K}  \Bigl[ \mathbb{E}_{p^{\frac{1-k}{2-k}} V_1(p)^{1-k}} \Bigl[ e^{\zeta T_K'(p)} \Big| p^{\frac{1-k}{2-k}} V_1(p)^{1-k} > K \Bigr] \Bigr], \\
& = e^{\zeta} + e^{\zeta}\mathbb{E}_{\pi_p^K}   \Bigl[ e^{ \lambda p^{\frac{1-k}{2-k}} V_1(p)^{1-k}} \mathbb{E}_{p^{\frac{1-k}{2-k}} V_1(p)^{1-k}} \Bigl[ e^{\zeta T'_K(p)} \\
& \hspace*{3.0cm}  e^{- \lambda p^{\frac{1-k}{2-k}} V_1(p)^{1-k} }  \Big|  p^{\frac{1-k}{2-k}} V_1(p)^{1-k} > K \Bigr] \Bigr]
\end{split}
\end{equation}
Now from Lemma \ref{lemma:hitting_time}, with appropriate choice of $K$, $\zeta$, and $\lambda > 0$ we have,
\begin{equation}
\begin{split}
\mathbb{E}_{\pi_p^K} \Bigl[ e^{\zeta T_K(p)} \Bigr] & \leq  e^{\zeta} + e^{\zeta} \mathbb{E}_{\pi_p^K}   \Bigl[ e^{ \lambda p^{\frac{1-k}{2-k}} V_1(p)^{1-k}} \Bigr], \\
& =  e^{\zeta} + e^{\zeta} \mathbb{E}_{\pi_p^K}   \Bigl[ e^{ \lambda p^{\frac{1-k}{2-k}} (1 - \beta) \Bigl(V_0(p)^{1-k} + \alpha (1 - k) G_{V_0(p)}^p\Bigr) } \Bigr], \\
& \leq e^{\zeta} +  e^{\zeta} e^{ \lambda (1 - \beta) K} \mathbb{E}_{\pi_p^K}   \Bigl[ e^{ \lambda p^{\frac{1-k}{2-k}} \alpha (1 - k) G_{V_0(p)}^p } \Bigr],  
\end{split}
\end{equation}
since $p^{\frac{1-k}{2-k}} V_0(p)^{1-k} \in [0,K]$. Using Proposition \ref{prop:Gbarx_CTCP1}, with a choice of small enough $\lambda$, we see that $\mathbb{E}_{\pi_p^K} \Bigl[ e^{\zeta T_K(p)} \Bigr]$ is bounded for $p \in (0,1)$ for some $ \zeta > 0$. Thus with initial distribution $\pi_p^K$, the random variables $T_K(p)$ are uniformly integrable. From Proposition \ref{prop:convergence_hitting_times},  
\begin{equation}
\mathbb{P}_{\pi_p^K}(T_K(p) = n) \rightarrow \mathbb{P}_{\overline{\pi}^K}(\overline{T}_K = n),
\end{equation}
as $p \rightarrow 0$. Therefore, we have
\begin{equation}
\label{eqn:stat_dist_conv_Dr_WC}
\lim_{p \rightarrow 0} \mathbb{E}_{\pi_p^K}  [T_K(p)] =  \mathbb{E}_{\overline{\pi}^K}  [\overline{T}_K].
\end{equation}
Since $T_K(p)$ is uniformly integrable, given $\epsilon > 0$, there exists $C$ such that 
\begin{equation}
\mathbb{E}_{\pi_p^K}  [T_K(p) \mathds{1}_{T_K(p) > C}] \leq \epsilon, 
\end{equation}
for all $p \in (0,1)$. Also,
\begin{equation}
\mathbb{E}_{\overline{\pi}^K}  [\overline{T}_K \mathds{1}_{\overline{T}_K > C}] \leq \epsilon.
\end{equation}
Now using the above and the uniform convergence in equation \eqref{eqn:Vbarx_CTCP2} of Proposition \ref{prop:Vbarx_CTCP}, for any bounded continuous function $f$, we have 
\begin{equation}
\label{eqn:stat_dist_conv_Nr_WC}
\lim_{p \rightarrow 0} \mathbb{E}_{\pi_p^K} \Bigl[  \displaystyle{\sum\limits_{n =0}^{T_{K}(p) - 1}} f(p^{\frac{1-k}{2-k}} V_n(p)^{1-k} ) \Bigr] = \mathbb{E}_{\overline{\pi}^K} \Bigl[  \displaystyle{\sum\limits_{n =0}^{\overline{T}_{K} - 1}} f(\overline{V}_n^{1-k} ) \Bigr]
\end{equation} 
From \eqref{eqn:stat_dist_conv_Dr_WC} and \eqref{eqn:stat_dist_conv_Nr_WC}, we see that for any bounded continuous function $f$,
\begin{equation}
\lim_{p \rightarrow 0} \mathbb{E}_{\pi_p}[f] = \mathbb{E}_{\overline{\pi}}[f],
\end{equation}
this shows that the distribution $\pi_p$ converges weakly to $\overline{\pi}$ as $p \rightarrow 0$, which proves the desired result.
\end{proof}

Let $\mathbb{E}[W(p)]$ denote the time average window size. Let $V_{\infty}(p)$ is a random variable which has the same distribution as the invariant distribution of $V_n$ process. From Palm calculus, we have
\begin{equation}
\label{eqn:EW_Dumas_approx}
\mathbb{E}[W(p)] = \frac{1}{p \mathbb{E}[G_{V_{\infty}(p)}^p]},
\end{equation}
since the average number of packets sent between losses is $\frac{1}{p}$. Let $\overline{V}_{\infty} $ be a random variable which has the same distribution as the stationary distribution of the $\overline{V}_n$ process. The following result, gives us an approximation to the time average window size of TCP Compound for small packet error rates.
\begin{theorem}
\begin{equation}
\lim_{p \rightarrow 0} \mathbb{E}[p^{\frac{1}{2-k}} W(p)] = \frac{1}{\mathbb{E}[\overline{G}_{\overline{V}_{\infty}}]}.
\end{equation}
\end{theorem}
\begin{proof}
Using \eqref{eqn:EW_Dumas_approx}, we get
\begin{equation}
\mathbb{E}[p^{\frac{1}{2-k}} W(p)] = \frac{1}{\mathbb{E}[p^{\frac{1-k}{2-k}} G_{V_{\infty}(p)}^p]},
\end{equation}
We shall now prove that $p^{\frac{1-k}{2-k}} G_{V_{\infty}(p)}^p$  converges weakly to $\overline{G}_{\overline{V}_{\infty}}$.
\begin{equation}
\mathbb{P}( p^{\frac{1-k}{2-k}} G_{V_{\infty}(p)}^p \geq y ) = \sum_{x = p^{\frac{1}{2-k}}, 2p^{\frac{1}{2-k}}, \cdots} \mathbb{P}( p^{\frac{1-k}{2-k}} G_{\frac{x}{p^{\frac{1}{2-k}}}}^p \geq y ) \mathbb{P}(p^{\frac{1}{2-k}} V_{\infty}(p) = x).
\end{equation}
Let $M < \infty$ be such that $\mathbb{P}(\overline{V}_{\infty}\geq M) < \epsilon$. Since, $p^{\frac{1}{2-k}} V_{\infty}(p)$ converges weakly to $ \overline{V}_{\infty}$, for small enough $p$, we have $\mathbb{P}(p^{\frac{1}{2-k}}V_{\infty}(p) \geq M) \leq \epsilon$. Therefore,  for small enough $p$,
\begin{equation}
\mathbb{P}( p^{\frac{1-k}{2-k}} G_{V_{\infty}(p)}^p \geq y ) \approx_{\epsilon} \sum_{x = p^{\frac{1}{2-k}}}^{M}  \mathbb{P}( p^{\frac{1-k}{2-k}} G_{\frac{x}{p^{\frac{1}{2-k}}}}^p \geq y ) \mathbb{P}(p^{\frac{1}{2-k}} V_{\infty}(p) = x).
\end{equation}
Using the uniform convergence of $\mathbb{P}(p^{\frac{1-k}{2-k}} G_{\floor{\frac{x}{p^{\frac{1}{2-k}}}}}^p \geq y)$ to $\mathbb{P}(\overline{G}_x  \geq y)$ over the compact interval $[0,M]$ (see \eqref{eqn:Gbarx_CTCP3}), we have
\begin{equation}
\label{eqn:convergence_avg_Gx}
\begin{split}
\mathbb{P}( p^{\frac{1-k}{2-k}} G_{V_{\infty}(p)}^p \geq y ) &\approx_{2 \epsilon} \sum_{x = p^{\frac{1}{2-k}}}^{M}  \mathbb{P}( \overline{G}_x  \geq y ) \mathbb{P}(p^{\frac{1}{2-k}} V_{\infty}(p) = x), \\
&\approx_{3 \epsilon}  \sum_{x = p^{\frac{1}{2-k}}, 2p^{\frac{1}{2-k}}, \cdots }  \mathbb{P}( \overline{G}_x  \geq y ) \mathbb{P}(p^{\frac{1}{2-k}} V_{\infty}(p) = x),
\end{split}
\end{equation}
for all $p \leq p^*$ for some $p^* > 0$. Since, $\mathbb{P}(\overline{G}_x  \geq y)$ is a continuous, bounded function of $x$, the RHS of \eqref{eqn:convergence_avg_Gx} converges to $\int \mathbb{P}( \overline{G}_x  \geq y ) d\mathbb{P}_{\overline{V}_{\infty}}(x)$ as $p \rightarrow 0$.  Therefore, we get
\begin{equation}
\begin{split}
\lim_{p \rightarrow 0} \mathbb{P}( p^{\frac{1-k}{2-k}} G_{V_{\infty}(p)}^p \geq y ) &= \int \mathbb{P}( \overline{G}_x  \geq y ) d\mathbb{P}_{\overline{V}_{\infty}}(x),\\
&= \mathbb{P}( \overline{G}_{\overline{V}_{\infty}}  \geq y ).
\end{split}
\end{equation}
From Proposition \ref{prop:Gbarx_CTCP1}, $p^{\frac{1-k}{2-k}} G_{V_{\infty}(p)}^p$ is uniformly integrable. Therefore,
\begin{equation}
\lim_{p \rightarrow 0} \mathbb{E}( p^{\frac{1-k}{2-k}} G_{V_{\infty}(p)}^p) = \mathbb{E}( \overline{G}_{\overline{V}_{\infty}}).
\end{equation}
This proves the desired result.
\end{proof}
We can evaluate $\mathbb{E}[\overline{G}_{\overline{V}_{\infty}}]$ using Monte-Carlo simulations. In Figure \ref{fig:CTCP_Gbar_time_avg}, we illustrate simulations for evaluating $\sum_{i=1}^{n} \frac{\overline{G}_{\overline{V}_{i}}}{n}$ with initial conditions $\overline{V}_{0} =$, $0.0$, $0.1$, $2.0$ for TCP Compound with parameters used in \cite{Tan2006Infocom}. We see that in these cases, after $n > 100$, there is little change in $\sum_{i=1}^{n} \frac{\overline{G}_{\overline{V}_{i}}}{n}$. For the TCP Compound parameters as used above and using $n = 10000$, we get $\mathbb{E}[\overline{G}_{\overline{V}_{\infty}}] \approx 3.9002$. 
\begin{figure}
\centering
\includegraphics[scale=0.22]{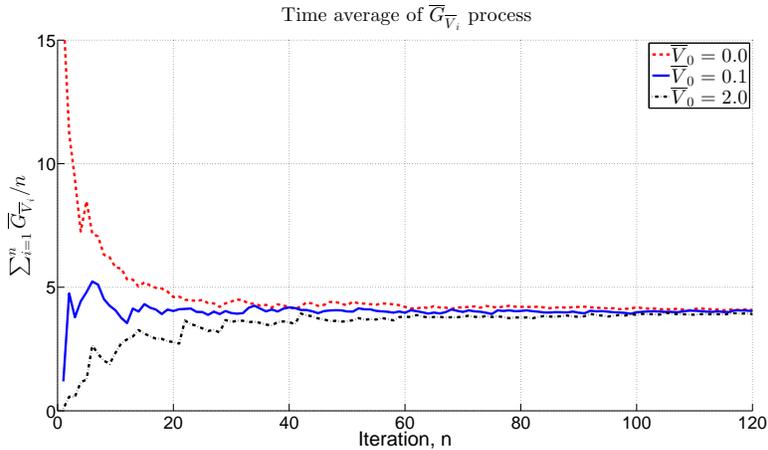}
\caption{Monte-Carlo simulations for $\mathbb{E}[\overline{G}_{\overline{V}_{\infty}}] \approx \sum_{i=1}^{n} \frac{\overline{G}_{\overline{V}_{i}}}{n}$}
\label{fig:CTCP_Gbar_time_avg}
\end{figure}
Therefore for small $p$, we can approximate the average window size as follows
\begin{equation}
\label{eqn:CTCP_avg_win_size}
\mathbb{E}[W(p)] = 0.2564p^{-\frac{1}{2-k}},
\end{equation}
and the throughput is given by $\lambda(p) = \frac{\mathbb{E}[W(p)]}{R}$.

An advantage of \eqref{eqn:CTCP_avg_win_size} over the results in \cite{Blanc2009} and \cite{Sudheer2013} is that now for given parameters, $\alpha$, $\beta$ of TCP Compound with one Monte-Carlo simulation we have a closed form expression of the average window size as a function of $p$ for all small $p$. It also shows the explicit dependence on $p$ and $k$.
\section{Simulation Results}
\label{sec:simulation_results}
In Tables \ref{tbl:ctcp_dumas_EW} and \ref{tbl:ctcp_dumas_G}, we compare the approximation obtained in Section \ref{sec:markovchain} with the fluid approximate model from \cite{Tan2006Infocom} (equation \eqref{eqn:compound_rf}) and our earlier Markov model in \cite{Sudheer2013}. In the ns2 simulations, the RTT was set to $0.1$ sec, the link speed was set to $10$ Gbps so that there is negligible queuing. The packet sizes were set to $1050$ bytes, the ns2 default value. For TCP Compound with fixed RTT and negligible queuing, the average window size, for the fluid model as well as the Markov chain model does not depend on the RTT. We see that the results obtained using the fluid model as well as the approximation developed in this paper show a good match with simulation. When compared against the ns2 simulations, all the models have similar accuracy. Most errors are less than $4\%$ for all the models. The Markov chain model of \cite{Sudheer2013} is marginally better.
\begin{table}
\centering
\caption{Average window size via different approximations}
\begin{tabular}{|c|c|c|c|c|}
\hline
per & $\mathbb{E}[W]$ & $\mathbb{E}[W]$ & $\mathbb{E}[W]$ & $\mathbb{E}[W]$ \\
\hline
$(p)$ & Simulations  & Det. Fluid & Markov chain & Approx. Markov \\
& (ns2) & \cite{Tan2006Infocom} & \cite{Sudheer2013} & $0.2570p^{-\frac{1}{2-k}}$ \\
\hline
$\num{1e-2}$ & $13.06$ & $10.16$ & $12.99$ & $10.23$ \\ 
\hline
$\num{5e-3}$ & $19.16$ & $17.69$ & $19.06$ & $17.81$ \\ 
$\num{3e-3}$ & $25.83$ & $26.63$ & $26.67$ & $26.80$ \\ 
$\num{1e-3}$ & $58.96$ & $64.12$ & $63.46$ & $64.54$ \\ 
\hline
$\num{8e-4}$ & $71.67$ & $76.66$ & $76.16$ & $77.16$ \\ 
$\num{5e-4}$ & $108.30$ & $111.65$ & $111.68$ & $112.38$ \\ 
$\num{3e-4}$ & $166.62$ & $168.01$ & $168.90$ & $169.11$ \\ 
$\num{1e-4}$ & $414.97$ & $404.60$ & $409.26$ & $407.25$ \\ 
\hline
$\num{8e-5}$ & $499.14$ & $483.67$ & $489.64$ & $486.85$ \\ 
$\num{5e-5}$ & $727.09$ & $704.45$ & $714.25$ & $709.07$ \\ 
$\num{3e-5}$ & $1115.37$ & $1060.05$ & $1076.51$ & $1067.01$ \\ 
\hline
\end{tabular}
\label{tbl:ctcp_dumas_EW}
\end{table}

\begin{table}
\centering
\caption{Goodput via different approximations}
\begin{tabular}{|c|c|c|c|c|}
\hline
per & Goodput & Goodput & Goodput & Goodput \\
\hline
$(p)$ & Simulations  & Det. Fluid & Markov chain & Approx. Markov \\
& (ns2) & \cite{Tan2006Infocom} & \cite{Sudheer2013} & $\frac{(1-p)0.2570p^{-\frac{1}{2-k}}}{R}$ \\
\hline
$\num{1e-2}$ & $128.84$ & $100.61$ & $128.61$ & $101.27$ \\
\hline
$\num{5e-3}$ & $189.94$ & $176.06$ & $189.62$ & $177.22$ \\
$\num{3e-3}$ & $257.36$ & $265.47$ & $265.92$ & $267.22$ \\
$\num{1e-3}$ & $588.88$ & $640.60$ & $633.95$ & $644.81$ \\
\hline
$\num{8e-4}$ & $715.88$ & $765.96$ & $761.02$ & $770.98$ \\
$\num{5e-4}$ & $1082.15$ & $1115.92$ & $1116.22$ & $1123.24$ \\
$\num{3e-4}$ & $1665.00$ & $1679.57$ & $1688.53$ & $1690.58$ \\
$\num{1e-4}$ & $4146.45$ & $4045.58$ & $4092.16$ & $4072.12$ \\
\hline
$\num{8e-5}$ & $4986.78$ & $4836.35$ & $4896.05$ & $4868.07$ \\
$\num{5e-5}$ & $7262.94$ & $7044.12$ & $7142.10$ & $7090.32$ \\
$\num{3e-5}$ & $11137.35$ & $10600.21$ & $10764.81$ & $10669.74$ \\
\hline
\end{tabular}
\label{tbl:ctcp_dumas_G}
\end{table}
\section{Conclusion}
\label{sec:conclusion}
In this paper, we have derived expressions for throughput of TCP Compound connections under random losses. We first looked at a deterministic loss model where a loss happens every $\frac{1}{p}$ packets, where $p$ is the packet error rate. We have shown that the window size process under this model has asymptotically periodic behaviour independent of the initial window size. This, theoretically justifies a fluid approximation for TCP Compound available in literature. We have then studied the system with random losses, which is more realistic, at least for wireless channels. We consider the window sizes at loss epochs. Using the insight from the fluid approximation, we have shown that the resulting Markov chains, indexed by $p$, when appropriately scaled converge to a limiting Markov chain as $p \rightarrow 0$. We have also shown that this Markov chain has a unique stationary distribution. This stationary distribution is then used to derive a closed from expression for TCP Compound average window size. We have compared our results with ns2 simulations and observed a good match between theory and simulations.
\begin{appendix}
\section{Appendix}
\label{app:appendixA}
\begin{proposition}
\label{prop:mgf_bound}
Suppose $\{X(p)\}$ is a family of positive random variables indexed by $p \in I$, for some $I \subset \mathbb{R}$ such that 
\begin{equation}
\mathbb{P}(X(p) \geq y) \leq c_1 \exp(-c_2 y),
\end{equation}
for all $y$, with $c_1 > 0$ and $c_2 > 0$. Then there exists $t>0$ such that for all $u \in (-t,t)$
\begin{equation}
\sup_{p \in I} \mathbb{E}[e^{uX(p)}] < \infty.
\end{equation}
\end{proposition}
\begin{proof}
For any $u > 0$ and $p > 0$, 
\begin{equation}
\label{eqn:mgf_bound1}
\begin{split}
\mathbb{E}[e^{uX(p)}] &= \int_{0}^{\infty} \mathbb{P}(e^{uX(p)} \geq y) dy, \\
& \leq 1 + \int_{1}^{\infty} \mathbb{P}(uX(p) \geq \log(y)) dy, \\
& \leq 1 + \int_{1}^{\infty} c_1 \exp(- \frac{c_2}{u} \log(y)) dy, \\
& = 1 + \int_{1}^{\infty} \frac{c_1}{y^{\frac{c_2}{u}}} dy.
\end{split}
\end{equation}
For $u < c_2$, $\int_{1}^{\infty} \frac{c_1}{y^{\frac{c_2}{u}}} < \infty$. Therefore for all, $u \in [0,c_2)$, $\sup_{p \in I} \mathbb{E}[e^{uX(p)}] < \infty$. For $u<0$, $uX(p) \leq 0$ which implies $\sup_{p \in I} \mathbb{E}[e^{uX(p)}] \leq 1$, for $u < 0$. 
\end{proof}

\begin{proposition}
\label{prop:unif_conv}
Let $\{X_p(x), x \in \mathbb{R}^+\}$, be a process, for $0 < p < 1$, which converges to a limiting process $X(x)$ uniformly in the sense,
\begin{equation}
\label{eqn:unifconv_1}
\lim_{p \rightarrow 0} \sup_{x, y \leq M} \Bigl| \mathbb{P}(X_p(x) \leq y) - \mathbb{P}(X(x) \leq y) \Bigr| = 0,
\end{equation}
for any finite $M$, and for each $x$, the limiting distribution, $\mathbb{P}(X(x) \leq y) $ is continuous. Then,
\begin{equation}
\lim_{p \rightarrow 0} \sup_{x \leq M} \Bigl| \mathbb{E}f(X_p(x)) -  \mathbb{E}f(X(x)) \Bigr| = 0,
\end{equation}
for any $f:\mathbb{R^+}\rightarrow\mathbb{R}$ continuous with compact support.
\end{proposition}
\begin{proof}
Consider a continuous function $f$ with compact support, $[0,K]$. Such a function is uniformly continuous. Therefore, given any $\epsilon$, there exists $m$ points $u_0 = 0 < u_1 < \cdots <  u_m = K$, such that $$\sup_{u_i< y < u_{i+1}} |f(y) - f(u_i)| < \epsilon,$$ for all $i=1,2,\cdots, m$.
We have 
\begin{equation} 
\begin{split} 
E[f(X_p(x)] &= \int\limits_{0}^{K} f(u) \mathbb{P}(X_p(x) \in du) \\ 
&\approx_{\epsilon} \sum\limits_{i=1}^{m-1} f(u_i) \mathbb{P}(X_p(x) \in (u_i,u_{i+1}]), 
\end{split}
\end{equation} 
where $\approx_{\epsilon}$ indicates that the RHS and LHS are $\epsilon$-close to each other. Similarly,
\begin{equation} 
\begin{split} 
E[f(X(x)] \approx_{\epsilon} \sum\limits_{i=1}^{m-1} f(u_i) \mathbb{P}(X(x) \in (u_i,u_{i+1}]). 
\end{split}
\end{equation}
Therefore, 
\begin{equation} 
\begin{split} 
\Bigl|\mathbb{E}f(X_p(x) - \mathbb{E}f(X(x)) \Bigr| \approx_{2\epsilon} & \sum\limits_{i=1}^{m} f(u_i) \\
&\Bigl| \mathbb{P}(X_p(x) \in (u_i,u_{i+1}]) - \mathbb{P}(X(x) \in  (u_i,u_{i+1}])\Bigr|,  \\
\leq & \sum\limits_{i=1}^{m}  \parallel f \parallel_{\infty} \\ 
&|\mathbb{P}(X_p(x) \in (u_i,u_{i+1}]) - \mathbb{P}(X(x) \in (u_i,u_{i+1}])|,
\end{split}
\end{equation} 
where $\parallel f \parallel_{\infty} = \sup \{f(x): x \in [0,K] \}$. Since $f$ is continuous over a compact support, it is bounded and hence $ \parallel f \parallel_{\infty} < \infty$.
Therefore
\begin{equation} 
\begin{split} 
\lim_{p \rightarrow 0} & \sup_x \Bigl|\mathbb{E}f(X_p(x) - \mathbb{E}f(X(x)) \Bigr|  \\ 
& \leq  \lim_{p \rightarrow 0} \parallel f \parallel_{\infty} \sum\limits_{i=1}^{m} \sup_x \Bigl| \mathbb{P}_{p,x}((u_i,u_{i+1}]) - \mathbb{P}_{x}((u_i,u_{i+1}]) \Bigr| + 2 \epsilon, \\
& = 2 \epsilon.
\end{split}
\end{equation} 
The second inequality follows from the hypothesis \eqref{eqn:unifconv_1}. Since $\epsilon$ is arbitrary we get the desired result.
\end{proof}

\end{appendix}

\bibliographystyle{IEEEtran} 
\bibliography{tcp-references}
\end{document}